\documentclass[a4paper,leqno]{amsart}

\usepackage{latexsym}
\usepackage[english]{babel}
\usepackage{fancyhdr}
\usepackage[mathscr]{eucal}
\usepackage{amsmath}
\usepackage{mathrsfs}
\usepackage{amsthm}
\usepackage{amsfonts}
\usepackage{amssymb}
\usepackage{amscd}
\usepackage{bbm}
\usepackage{graphicx}
\usepackage{graphics}
\usepackage{latexsym}
\usepackage{color}

\newcommand{\ud}{\mathrm{d}}

\newcommand{\ii}{\mathrm{i}}

\theoremstyle{plain}
\newtheorem{theorem}{Theorem}[section]
\newtheorem{lemma}[theorem]{Lemma}
\newtheorem{corollary}[theorem]{Corollary}
\newtheorem{proposition}[theorem]{Proposition}

\theoremstyle{definition}

\newtheorem{remark}[theorem]{Remark}
\newtheorem*{remark*}{Remark}

\numberwithin{equation}{section}

\begin{document}

\title[Hydrogenoid spectra with central perturbations]
{Hydrogenoid spectra with central perturbations}
\author[M.~Gallone]{Matteo Gallone}
\address[M.~Gallone]{International School for Advanced Studies -- SISSA \\ via Bonomea 265 \\ 34136 Trieste (Italy).}
\email{mgallone@sissa.it}
\author[A.~Michelangeli]{Alessandro Michelangeli}
\address[A.~Michelangeli]{International School for Advanced Studies -- SISSA \\ via Bonomea 265 \\ 34136 Trieste (Italy).}
\email{alemiche@sissa.it}


\begin{abstract}
Through the Kre\u{\i}n-Vi\v{s}ik-Birman extension scheme, unlike the classical analysis based on von Neumann's theory, we reproduce the construction and classification of all self-adjoint realisations of three-dimensional hydrogenoid-like Hamiltonians with singular perturbation supported at the Coulomb centre (the nucleus), as well as of Schr\"{o}dinger operators with Coulomb potentials on the half-line. These two problems are technically equivalent, albeit sometimes treated by their own in the the literature.
Based on such scheme, we then recover the formula to determine the eigenvalues of each self-adjoint extension, as corrections of the non-relativistic hydrogenoid energy levels. We discuss in which respect the Kre\u{\i}n-Vi\v{s}ik-Birman scheme is somewhat more natural in yielding the typical boundary condition of self-adjointness at the centre of the perturbation.
\end{abstract}

\date{\today}

\subjclass[2000]{}
\keywords{Quantum hydrogenoid Hamiltonians. Schr\"odinger-Coulomb on half-line. Self-adjoint extensions. Kre\u{\i}n-Vi\v{s}ik-Birman theory. Whittaker functions. Point interactions}

%

\maketitle

\tableofcontents

\section{Hydrogenoid Hamiltonians with point-like perturbation at the centre: outlook and main results}
\label{sec:intro}

We are concerned in this work, mainly from the perspective of the novelty of the approach to a problem otherwise classical in the literature, with certain realistic types of perturbations of the familiar quantum Hamiltonian for the valence electron of hydrogenoid atoms, namely the operator
\begin{equation}\label{eq:Hhydr}
H_{\mathrm{Hydr}} \; = \; -\frac{\hbar^2}{2m} \Delta - \frac{Z e^2}{|x|}
\end{equation} 
on $L^2(\mathbb{R}^3)$ with domain of self-adjointness $H^2(\mathbb{R}^3)$, where $m$ and $-e$ are, respectively, the electron's mass and charge ($e>0$), $Z$ is the atomic number of the nucleus, $\hbar$ is Planck's constant and $\Delta$ is the three-dimensional Laplacian.

In particular, we are concerned with the deviations from the celebrated spectrum of the hydrogen atom:
\begin{equation}\label{eq:spectrum_hydrogeonid}
 \begin{split}
 \sigma_{\mathrm{ess}}(H_{\mathrm{Hydr}})\;&=\;\sigma_{\mathrm{ac}}(H_{\mathrm{Hydr}})\;=\;[0,+\infty)\,,\quad\sigma_{\mathrm{sc}}(H_{\mathrm{Hydr}})\;=\;\varnothing \\
 \sigma_{\mathrm{point}}(H_{\mathrm{Hydr}})\;&=\;\Big\{-mc^2\frac{\,(Z\alpha_{\mathrm{f}})^2}{2 n^2}\,\Big|\,n\in\mathbb{N} \Big\}
 \end{split}
\end{equation}
where $\alpha_{\mathrm{f}}=\frac{\:e^2}{\hbar c}\approx\frac{1}{137}$ is the fine structure constant and $c$ is the speed of light.

Intimately related to this problem, we are concerned with the problem of the self-adjoint realisations of the `radial' differential operator
\begin{equation}
 -\frac{\ud^2}{\ud r^2}+\frac{\nu}{r}\,,\qquad \nu\in\mathbb{R}
\end{equation}
on the Hilbert space of the half-line, $L^2(\mathbb{R}^+,\ud r)$, and on the classification of all such realisations and the characterisation of their spectra.

In order to explain the scope of our study, approach, and results, let us discuss the following preliminaries.

\subsection{Fine structure and Darwin correction}~

As well known \cite[\S 34]{LandauLifshitz-4}, standard calculations within first-order perturbation theory, made first by Sommerfeld even before the complete definition of quantum mechanics, show that the correction $\delta E_n^{(\mathrm{H})}$ to the $n$-th eigenvalue $E_n^{(\mathrm{H})}:=-\frac{\,(Z\alpha_{\mathrm{f}})^2}{2 n^2}$ of \eqref{eq:spectrum_hydrogeonid} is given by
\begin{equation}\label{eq:1storderPert}
 \frac{\,\delta E_n^{(\mathrm{H})}}{E_n^{(\mathrm{H})}}\;=\;-\frac{\,(Z\alpha_{\mathrm{f}})^2}{n}\Big(\frac{1}{j+\frac{1}{2}}-\frac{3}{4n}\Big)\,,
\end{equation}
where $j$ is the quantum number of the total angular momentum, thus $j=\frac{1}{2}$ if $\ell=0$ and $j=\ell\pm\frac{1}{2}$ otherwise, in the standard notation that we shall remind in a moment. (The net effect is therefore a partial removal of the degeneracy of $E_n^{(\mathrm{H})}$ in the spin of the electron and in the angular number $\ell$, a double degeneracy remaining for levels with the same $n$ and $\ell=j\pm\frac{1}{2}$, apart from the maximum possible value $j_{\mathrm{max}}=n-\frac{1}{2}$.)

Let us recall (see, e.g., \cite[Chapter 6]{Thaller-Dirac-1992}) that the first-order perturbative scheme yielding \eqref{eq:1storderPert} corresponds to adding to $H_{\mathrm{Hydr}}$ corrections that arise in the non-relativistic limit from the Dirac operator for the considered atom: $H_{\mathrm{Hydr}}$ is indeed formally recovered as one of the two identical copies of the spinor Hamiltonian obtained from the Dirac operator as $c\to +\infty$, and the eigenvalues of the latter, once the rest energy $mc^2$ is removed, converge to those of $H_{\mathrm{Hydr}}$, with three types of subleading corrections, to the first order in $1/c^{2}$:
\begin{itemize}
 \item the \emph{kinetic energy correction}, interpreted in terms of the replacement of the relativistic with the non-relativistic energy, that classically amounts to the contribution
 \[
  \Big( \sqrt{c^2p^2-m^2c^4}-mc^2\Big)-\frac{\,p^2}{2m}\;=\;-\frac{1}{8m^2c^2}\,p^4+O(c^{-4})\,;
 \]
 \item the \emph{spin-orbit correction}, interpreted in terms of the interaction of the magnetic moment of the electron with the magnetic field generated by the nucleus in the reference frame of the former, including also the effect of the Thomas precession;
 \item the \emph{Darwin term correction}, interpreted as an effective smearing out of the electrostatic interaction between the electron and nucleus due to the Zitterbewegung, the rapid quantum oscillations of the electron.
\end{itemize}
In fact, each such modified eigenvalue $E_n^{(\mathrm{H})}+\delta E_n^{(\mathrm{H})}$ is the first-order term of the expansion in powers of $1/c^2$ of $E_{n,j}-mc^2$, where $E_{n,j}$ is the Dirac operator's eigenvalue given by Sommerfeld's celebrated \emph{fine structure formula}
\begin{equation}
 E_{n,j}\;=\;mc^2\Big(1+\frac{(Z\alpha_{\mathrm{f}})^2}{\:(n-j-\frac{1}{2}+\sqrt{\kappa^2-(Z\alpha_{\mathrm{f}})^2})^2}\Big)^{\!-\frac{1}{2}}\,.
\end{equation}

Let us recall, in particular, the nature of the Darwin correction, which is induced by the interaction between the magnetic moment of the moving electron and the electric field $\mathbf{E}=\frac{1}{e}\nabla V$, where $V$ is the potential energy due to the charge distribution that generates $\mathbf{E}$. 
This effect, to the first order in perturbation theory, produces an additive term to the non-relativistic Hamiltonian, which formally reads \cite[\S 33]{LandauLifshitz-4} 
\begin{equation}\label{eq:HDarw1}
 H_{\mathrm{Darwin}}\;=\;-\frac{\hbar^2}{8 m^2 c^2}\,e\,\mathrm{div}\mathbf{E}\;=\;-\frac{\hbar^2}{8 m^2 c^2}\,\Delta V\,.
\end{equation}
For a hydrogenoid atom $V(x)=-Ze^2/|x|$, whence $\Delta V=-4\pi Z e^2\delta^{(3)}(x)$: the term \eqref{eq:HDarw1} is therefore to be regarded as a \emph{point-like perturbation} `supported' at the centre of the atom, whose nuclear charge creates the field $\mathbf{E}$. 
In this case one gives meaning to \eqref{eq:HDarw1} in the sense of the expectation 
\begin{equation}\label{HDarwexp}
 \langle\psi,H_{\mathrm{Darwin}}\psi\rangle\;=\;\frac{4\pi Z e^2\hbar^2}{8 m^2 c^2}\,|\psi(0)|^2\;=\;E_n^{(\mathrm{H})}\,\frac{(Z\alpha_{\mathrm{f}})^2}{n}\cdot\pi\Big(\frac{n a_0}{Z}\Big)^{\!3}\,|\psi(0)|^2\,,
\end{equation}
where $a_0=\frac{\:\hbar^2}{me^2}$ is the Bohr radius.

Unlike the semi-relativistic kinetic energy and spin-orbit corrections, the Darwin correction only affects the $s$ orbitals ($\ell=0$, $j=\frac{1}{2}$), the wave functions of higher orbitals vanishing at $x=0$. Since the $s$-wave normalised eigenfunction $\psi_n^{(\mathrm{H})}$ corresponding to $E_n^{(\mathrm{H})}$ satisfies $|\psi_n^{(\mathrm{H})}\!(0)|^2=\frac{1}{\pi}(\frac{Z}{n a_0})^3$, \eqref{HDarwexp} implies 
\begin{equation}\label{eq:1storderPert-Darwin}
 \Big(\frac{\,\delta E_n^{(\mathrm{H})}}{E_n^{(\mathrm{H})}}\Big)_{\mathrm{Darwin}}\;=\;\frac{\,(Z\alpha_{\mathrm{f}})^2}{n}\qquad\qquad (\ell=0)\,.
\end{equation}

\subsection{Point-like perturbations supported at the interaction centre}~

The above classical considerations are one of the typical motivations for the rigorous study of a `simplified fine structure', low-energy correction of the ideal (non-relativistic) hydrogenoid Hamiltonian \eqref{eq:Hhydr} that consists of a Darwin-like perturbation only. In particular, one considers an additional interaction that is only present in the $s$-wave sector.

This amounts to constructing self-adjoint Hamiltonians with Coulomb plus point interaction centred at the origin, and it requires to go beyond the formal perturbative arguments that yielded the spectral correction \eqref{eq:1storderPert-Darwin}.

One natural approach, exploited first in the early 1980's works by Zorbas \cite{Zorbas-1980}, by   Albeverio, Gesztesy, H{\o}egh-Krohn, and Streit \cite{AGHKS-1983_Coul_plus_delta}, and by Bulla and Gesztesy \cite{Bulla-Gesztesy-1985}, is to regard such Hamiltonians as self-adjoint extensions of the densely defined, symmetric, semi-bounded from below operator
\begin{equation}\label{eq:Hring1}
  \mathring{H}_{\mathrm{Hydr}}\;=\;\Big(-\frac{\:\hbar^2}{2m}\Delta-\frac{Z e^2}{|x|}\Big)\Big|_{C^\infty_0(\mathbb{R}^3\setminus \{0\})}\,.
\end{equation}

For clarity of presentation we shall set $\nu:=-Ze^2$, in fact allowing $\nu$ to be positive or negative real, and we shall work in units $2m=\hbar=e=1$. We shall then write $H^{(\nu)}$ and $\mathring{H}^{(\nu)}$ for the operator $-\Delta+\frac{\nu}{|x|}$ defined, respectively, on the domain of self-adjointness $H^2(\mathbb{R}^3)$ or on the restriction domain $C^\infty_0(\mathbb{R}^3\setminus \{0\})$.

As was found in \cite{Zorbas-1980,AGHKS-1983_Coul_plus_delta,Bulla-Gesztesy-1985}, the self-adjoint extensions of $\mathring{H}^{(\nu)}$ on $L^2(\mathbb{R}^3)$ at fixed $\nu$ form a one-parameter family $\{H_\alpha^{(\nu)}|\alpha\in(-\infty,+\infty]\}$ of rank-one perturbations, in the resolvent sense, of the Hamiltonian $H^{(\nu)}$. We state this famous result in Theorem \ref{thm:3Dclass} below.

In fact, in this work among other findings we shall \emph{re-obtain} such a result through an alternative path.
Indeed, the above-mentioned works \cite{Zorbas-1980,AGHKS-1983_Coul_plus_delta,Bulla-Gesztesy-1985} the standard self-adjoint extension theory a la von Neumann \cite[Chapt.~8]{Weidmann-LinearOperatosHilbert} was applied. We intend to exploit here an alternative construction and classification based on the Kre{\u\i}n-Vi\v{s}ik-Birman extension scheme \cite{GMO-KVB2017}, owing to certain features of the latter theory that are somewhat more informative and cleaner, in the sense that we are going to specify in due time.

We also recall that the integral kernel of $(H^{(\nu)}-k^2\mathbbm{1})^{-1}$ is explicitly known \cite{Hostler-1967}:
\begin{equation}
 \begin{split}
  (H^{(\nu)}-k^2\mathbbm{1})^{-1}(x,y)\;=&\,\;\Gamma({\textstyle1+\frac{\ii\nu}{2k}})\,\frac{\,\mathscr{A}_{\nu,k}(x,y)}{4\pi|x-y|}\,,\qquad x,y\in\mathbb{R}^3\,,\;\; x\neq y \\
  \mathscr{A}_{\nu,k}(x,y)\;:=&\;\Big(\frac{\ud}{\ud \xi}-\frac{\ud}{\ud \eta}\Big)\,\mathscr{M}_{-\frac{\ii\nu}{2k},\frac{1}{2}}(\xi)\,\mathscr{W}_{-\frac{\ii\nu}{2k},\frac{1}{2}}(\eta)\bigg|_{\substack{\xi=-\ii k z_- \\ \eta=-\ii k z_+}} \\
  z_\pm\;:=\;|x|+|y|&\pm|x-y|\,,\qquad k^2\in\rho(H_\alpha^{(\nu)})\,,\quad \mathfrak{Im}k>0\,,
 \end{split}
\end{equation}
 where $\mathscr{M}_{a,b}$ and $\mathscr{W}_{a,b}$ are the Whittaker functions \cite[Chapt.~13]{Abramowitz-Stegun-1964}.

%
%

\subsection{Angular decomposition}\label{sec:angular_decomp}~

Let us exploit as customary the rotational symmetry of $H^{(\nu)}$ and $\mathring{H}^{(\nu)}$ by passing to polar coordinates $x\equiv(r,\Omega)\in\mathbb{R}^+\!\times\mathbb{S}^2$, $r:=|x|$, for $x\in\mathbb{R}^3$. This induces the standard isomorphism
\begin{equation}\label{eq:ang_decomp}
\begin{split}
 L^2(\mathbb{R}^3,\ud x)\;&\cong\;
 U^{-1}L^2(\mathbb{R}^+,\ud r)\otimes L^2(\mathbb{S}^2,\ud\Omega) \\
 &\cong\;\bigoplus_{\ell=0}^\infty \Big( U^{-1}L^2(\mathbb{R}^+,\ud r)\otimes\mathrm{span}\{Y_{\ell}^{-\ell},\dots,Y_{\ell}^\ell\}\Big)
\end{split}
\end{equation}
where $U:L^2(\mathbb{R}^+,r^2\ud r)\to L^2(\mathbb{R}^+,\ud r)$ is the unitary $(Uf)(r)=rf(r)$, and the $Y^m_\ell$'s are the spherical harmonics on $\mathbb{S}^2$, i.e., the common eigenfunctions of $\boldsymbol{L}^2$ and $\boldsymbol{L}_3$ of eigenvalue $\ell(\ell+1)$ and $m$ respectively,  $\boldsymbol{L}=x \times (- \ii \nabla)$ being the angular momentum operator.

Standard arguments show that $\mathring{H}^{(\nu)}$ (and analogously $H^{(\nu)}$) is reduced by the decomposition \eqref{eq:ang_decomp} as
\begin{equation}\label{eq:ang_decomp_operator}
 \mathring{H}^{(\nu)}\;\cong\;\bigoplus_{\ell=0}^\infty \Big( U^{-1} h^{(\nu)}_\ell U\otimes\mathbbm{1} \Big)
\end{equation}
where each $h^{(\nu)}_\ell$ is the operator on $L^2(\mathbb{R}^+,\ud r)$ defined by
\begin{equation}\label{eq:AngularOperator}
h_{\ell}^{(\nu)} \; := \; -\frac{\ud^2}{\ud r^2} +\frac{\,\ell(\ell+1)}{r^2}+\frac{\nu}{r} \,,\qquad \mathcal{D}\big(h_{\ell}^{(\nu)}\big) \; :=\; C^\infty_0(\mathbb{R}^+)\,.
\end{equation}

\subsection{The radial problem}~

Owing to \eqref{eq:ang_decomp}-\eqref{eq:ang_decomp_operator}, the question of the self-adjoint extensions of $\mathring{H}^{(\nu)}$ on $L^2(\mathbb{R}^3,\ud x)$ is the same as the question of the self-adjoint extensions of each $h_{\ell}^{(\nu)}$ on $L^2(\mathbb{R}^+)$.

Based on the classical analysis of Weyl (see, e.g., \cite[Theorem 15.10(iii)]{schmu_unbdd_sa}), all the block operators $h_{\ell}^{(\nu)}$ with $\ell\in\mathbb{N}$ are essentially self-adjoint, as they are both in the limit point case at infinity \cite[Prop.~15.11]{schmu_unbdd_sa} and in the limit point case at zero \cite[Prop.~15.12(i)]{schmu_unbdd_sa}.

One could also add (but we shall retrieve this conclusion along a different path) that  $h_{0}^{(\nu)}$ is still in the limit point case at infinity, yet limit circle at zero \cite[Prop.~15.12(ii)]{schmu_unbdd_sa}, thus, admitting a one-parameter family of self-adjoint extensions \cite[Theorem 15.10(ii)]{schmu_unbdd_sa}.

The question of the self-adjoint realisations of $\mathring{H}^{(\nu)}$ is then boiled down to the self-adjointness problem for $h_{0}^{(\nu)}$ on $L^2(\mathbb{R}^+,\ud r)$.

This too is a problem studied since long, that we want to re-consider from an alternative, instructive perspective.

The first analysis in fact dates back to Rellich \cite{Rellich-1944} (even though self-adjointness was not the driving notion back then) and is based on Green's function methods to show that $-\frac{\ud^2}{\ud r^2} +\frac{\nu}{r}+\ii\mathbbm{1}$ is inverted by a bounded operator on Hilbert space when the appropriate boundary condition at the origin is selected. Some four decades later Bulla and Gesztesy \cite{Bulla-Gesztesy-1985} (a concise summary of which may be found in \cite[Appendix D]{albeverio-solvable}) produced a `modern' classification based on the special version of von Neumann's extension theory for second order differential operators \cite[Chapt.~8]{Weidmann-LinearOperatosHilbert}, in which the extension parameter that labels each self-adjoint realisation governs a boundary condition at zero analogous to \eqref{eq:bc}. (We already mentioned that the work \cite{Bulla-Gesztesy-1985} came a few years after Zorbas \cite{Zorbas-1980} and Albeverio, Gesztesy, H\o{}egh-Krohn, and Streit \cite{AGHKS-1983_Coul_plus_delta} had classified the self-adjoint realisations of the three-dimensional problem directly, i.e., without explicitly working out the reduction discusses in Sec.~\ref{sec:angular_decomp}.) More recently Gesztesy and Zinchenko \cite{Gesztesy-Zinchenko-2006} extended the scope of \cite{Bulla-Gesztesy-1985} to more singular potentials than $r^{-1}$.

The novelty of the present analysis, as we shall see, besides the explicit qualification of the closure and of the Friedrichs extension of $h_{0}^{(\nu)}$, is the relatively straightforward application of the alternative extension scheme of Kre{\u\i}n, Vi\v{s}ik, and Birman.


\subsection{Main results}~

Let us finally come to our main results. On the one hand, as mentioned already, we reproduce classical facts (namely Theorem \ref{thm:1Dclass} for the radial problem and Theorem \ref{thm:3Dclass} for the singularly-perturbed hydrogenoid Hamiltonians) through the alternative extension scheme of Kre{\u\i}n, Vi\v{s}ik, and Birman. On the other hand, we qualify previously studied objects in an explicit, new form, specifically the Friedrichs realisation of the radial operator (Theorem \ref{thm:1Fried}) and our final formula for the central perturbation of the hydrogenoid spectra (Theorem \ref{thm:EV_corrections}).

Clearly, whereas the derivatives in \eqref{eq:Hring1} and \eqref{eq:AngularOperator} are \emph{classical}, the following formulas contain \emph{weak} derivatives.

%

As a first step, we identify the closure and the Friedrichs realisation of the radial problem.

\begin{theorem}[Closure and Friedrichs extension of $h^{(\nu)}_0$]\label{thm:1Fried}~

\noindent The operator $h_0^{(\nu)}$ is semi-bounded from below with deficiency index one. 
\begin{itemize}
\item[(i)] One has
\begin{equation}
\begin{split}
 \mathcal{D}(\overline{h_0^{(\nu)}})\;&=\;H^2_0(\mathbb{R}^+)\;=\;\overline{C^\infty_0(\mathbb{R}^+)}^{\Vert \, \Vert_{H^2}} \\
 \overline{h_0^{(\nu)}}f\; &= \; - f''+\frac{\nu}{r} f \, .
\end{split}
\end{equation}
\end{itemize}
The Friedrichs extension $h_{0,F}^{(\nu)}$ of $h_0^{(\nu)}$ has 
\begin{itemize}
\item[(ii)] operator domain and action given by
 \begin{equation}\label{eq:DomainHF}
\begin{split}
\mathcal{D}(h_{0,F}^{(\nu)})\;&=\;H^2(\mathbb{R}^+) \cap H^1_0(\mathbb{R}^+)\;=\;\{f \in H^2(\mathbb{R}^+)\,|\, \lim_{r \downarrow 0} f(r) = 0\} 
\\
h_{0,F}^{(\nu)} f \; &= \; - f''+\frac{\nu}{r} f \, ;
\end{split}
\end{equation}
\item[(iii)] quadratic form given by
\begin{equation}\label{eq:DomainHFform}
\begin{split}
\mathcal{D}[h_{0,F}] & = H^1_0(\mathbb{R}^+) \\
h_{0,F}^{(\nu)}[f,h] \; &=\; \int_0^{+\infty} \Big(\overline{f'(r)} h'(r)+ \nu \frac{\overline{f(r)} h(r)}{r} \Big) \, \ud r \, ;
\end{split}
\end{equation}
\item[(iv)] resolvent with integral kernel
\begin{equation}\label{eq:DomainInverse}
\begin{split}
 \Big( h^{(\nu)}_{0,F}&+\frac{\nu^2}{4 \kappa^2} \Big)^{-1}(r,\rho)\;= \\
 &=\;-\frac{\kappa \Gamma(1-\kappa)}{\nu} \begin{cases}
\mathscr{W}_{\kappa,\frac{1}{2}}( -\textstyle{\frac{\nu}{\kappa}} r) \mathscr{M}_{\kappa,\frac{1}{2}}(-\textstyle{\frac{\nu}{\kappa}} \rho) \qquad \text{if }0 < \rho < r\\
\mathscr{M}_{\kappa,\frac{1}{2}}(-\textstyle{\frac{\nu}{\kappa}} r)\mathscr{W}_{\kappa,\frac{1}{2}}( -\textstyle{\frac{\nu}{\kappa}} \rho) \qquad \text{if } 0 < r < \rho\,,
\end{cases}
\end{split}
\end{equation}
where $\kappa \in (-\infty,0)\cup(0,1)$, $\mathrm{sign} \, \kappa = - \mathrm{sign} \, \nu$, and where $\mathscr{W}_{a,b}(r)$ and $\mathscr{M}_{a,b}(r)$ are the Whittaker functions.
\end{itemize}
\end{theorem}

Next, using the Friedrichs extension as a \emph{reference extension} for the Kre{\u\i}n-Vi\v{s}ik-Birman scheme, we classify all other self-adjoint realisations of the radial problem. The result is classical in the literature \cite{Rellich-1944,Bulla-Gesztesy-1985}, but we find the present derivation more straightforward and natural, especially in yielding the typical boundary condition at the origin that qualify each extension. 

\begin{theorem}[Self-adjoint realisations of $h^{(\nu)}_0$]~
\label{thm:1Dclass}

\begin{itemize}
\item[(i)] The self-adjoint extensions of $h_0^{(\nu)}$ form the family $(h_{0,\alpha}^{(\nu)})_{\alpha \in \mathbb{R} \cup \{\infty\}}$, where $\alpha=\infty$ labels the Friedrichs extension, and 
\begin{equation}\label{eq:dsa}
 \begin{split}
  \mathcal{D}(h^{(\nu)}_{0,\alpha})\;&=\;\left\{
  g\in L^2(\mathbb{R}^+)\left|
  \begin{array}{c}
   -g''+\textstyle{\frac{\nu}{r}g}\in L^2(\mathbb{R}^+) \\
   \textrm{and }\;g_1\;=\;4\pi\alpha\, g_0
  \end{array}\!
  \right.\right\} \\
  h^{(\nu)}_{0,\alpha}\,g\;&=\;-g''+\frac{\nu}{r}\,g\,,
 \end{split}
\end{equation}
$g_0$ and $g_1$ being the existing limits
\begin{equation}\label{eq:g0g1limits-statements}
 \begin{split}
  g_0\;&:=\;\lim_{r\downarrow 0}g(r) \\
  g_1\;&:=\;\lim_{r\downarrow 0}r^{-1}\big(g(r)-g_0(1+\nu r\ln r)\big)\,.
 \end{split}
\end{equation}
\item[(ii)] For given $\kappa \in (-\infty,0)\cup(0,1)$, $\mathrm{sign}\, \kappa = - \mathrm{sign}\, \nu$, one has
\begin{equation}\label{eq:1Dresolvent}
\Big(h^{(\nu)}_{0,\alpha})+ \frac{\nu^2}{4 \kappa^2} \Big)^{-1} \; = \; \Big( h^{(\nu)}_{0,\infty}+\frac{\nu^2}{4 \kappa^2} \Big)^{-1}+ \frac{\Gamma(1-\kappa)^2}{4 \pi} \frac{1}{\alpha-\mathfrak{F}_{\nu,\kappa}} | \Phi_\kappa \rangle \langle \Phi_\kappa |\,,
\end{equation}
where $\Phi_\kappa(r)\;:=\;\mathscr{W}_{\kappa,\frac{1}{2}}(-\frac{\nu}{\kappa} r)$ and 
\begin{equation}
\label{eq:Fnuk}
  \mathfrak{F}_{\nu,\kappa}\;:=\;\frac{\nu}{4\pi}\big(\psi(1-\kappa)+\ln(-{\textstyle\frac{\nu}{\kappa}})+(2\gamma-1)+{\textstyle\frac{1}{2 \kappa}}\big)\,.
\end{equation}
\end{itemize}
\end{theorem}

Consistently, when $\nu=0$ the boundary condition \eqref{eq:dsa} for the $\alpha$-extension takes the classical form $g'(0)=4\pi\alpha g(0)$, namely the well-known boundary condition for the generic self-adjoint Laplacian on the half-line \cite{Kostrykin-Schrader-2006,GTV-2012,DM-2015-halfline}.


When the radial analysis is lifted back to the three-dimensional Hilbert space, we re-obtain, through an alternative path, the following classification result already available in the literature (see, e.g., \cite[Theorem I.2.1.2]{albeverio-solvable}).

\begin{theorem}[Self-adjoint realisations of $\mathring{H}^{(\nu)}$]~
\label{thm:3Dclass}

\noindent The self-adjoint extensions of $\mathring{H}^{(\nu)}$ form the family $(H^{(\nu)}_\alpha)_{\alpha\in\mathbb{R}\cup\{\infty\}}$ characterised as follows. 
\begin{itemize}
 \item[(i)] With respect to the canonical decomposition \eqref{eq:ang_decomp} of $L^2(\mathbb{R}^3)$, the extension $H^{(\nu)}_\alpha$ is reduced as
 \begin{equation}\label{eq:ang_decomp_Halpha}
 H^{(\nu)}_\alpha\;\cong\;\bigoplus_{\ell=0}^\infty \Big( U^{-1} h^{(\nu)}_{\ell,\alpha}\, U\otimes\mathbbm{1} \Big)\,,
\end{equation}
where $h^{(\nu)}_{0,\alpha}$ is qualified in Theorem \ref{thm:1Dclass} and $h^{(\nu)}_{\ell,\alpha}$, for $\ell\geqslant 1$, is the closure of $h^{(\nu)}_{\ell}$ introduced in \eqref{eq:AngularOperator}, namely the $L^2(\mathbb{R}^+)$-self-adjoint operator
\begin{equation}\label{eq:domhell1}
 \begin{split}
  \mathcal{D}(h^{(\nu)}_{\ell,\alpha})\;&=\;\{g\in L^2(\mathbb{R}^+)\,|\,-g''+\textstyle{\frac{\,\ell(\ell+1)}{r^2}g+\frac{\nu}{r}g}\in L^2(\mathbb{R}^+)\} \\
  h^{(\nu)}_{\ell,\alpha}\,g\;&=\;-g''+\textstyle{\frac{\,\ell(\ell+1)}{r^2}g+\frac{\nu}{r}g}\,.
 \end{split}
\end{equation}
\item[(ii)] The choice $\alpha=\infty$ identifies the Friedrichs extension of $\mathring{H}^{(\nu)}$, which is precisely the self-adjoint hydrogenoid Hamiltonian
\begin{equation}\label{HnuFriedr}
 H^{(\nu)}\;=\;-\Delta+\frac{\nu}{\,|x|\,}\,,\qquad\mathcal{D}(H^{(\nu)})\;=\;H^2(\mathbb{R}^3)\,.
\end{equation}
 It is the only member of the family $(H^{(\nu)}_\alpha)_{\alpha\in\mathbb{R}\cup\{\infty\}}$ whose domain's functions have separately finite kinetic and finite potential energy, in the sense of energy forms.
 \item[(iii)] For given $\kappa \in (-\infty,0)\cup(0,1)$, $\mathrm{sign}\,\kappa=-\mathrm{sign}\,\nu$, one has
 \begin{equation}\label{eq:HnuResolvent}
  \Big(H^{(\nu)}_\alpha+\frac{\nu^2}{\,4\kappa^2}\,\mathbbm{1}\Big)^{\!-1}\;=\;\Big(H^{(\nu)}+\frac{\nu^2}{\,4\kappa^2}\,\mathbbm{1}\Big)^{\!-1} +\frac{1}{\,\alpha-\mathfrak{F}_{\nu,\kappa}\,}|\mathfrak{g}_{\nu,\kappa}\rangle\langle\mathfrak{g}_{\nu,\kappa}|\,,
 \end{equation}
 where
 \begin{equation}\label{eq:ourgnukappa}
 \mathfrak{g}_{\nu,\kappa}(x)\;:=\;\Gamma({\textstyle 1-\kappa})\,\frac{\,\,\mathscr{W}_{\kappa,\frac{1}{2}}(-\frac{\nu}{\kappa}|x|)}{4\pi|x|}
\end{equation}
and $\mathfrak{F}_{\nu,\kappa}$ is defined in \eqref{eq:Fnuk}.
  \item[(iv)] For given $\kappa \in (-\infty,0)\cup(0,1)$, $\mathrm{sign}\,\kappa=-\mathrm{sign}\,\nu$, one has
  \begin{equation}\label{eq:DomainHalpha}
   \begin{split}
    \mathcal{D}(H^{(\nu)}_\alpha)\;&=\;\Big\{ \psi=\varphi_\kappa+\frac{\varphi_\kappa(0)}{\,\alpha-\mathfrak{F}_{\nu,\kappa}\,}\,\mathfrak{g}_{\nu,\kappa}\Big|\,\varphi_\kappa\in H^2(\mathbb{R}^3)\Big\} \\
    \Big(H^{(\nu)}_\alpha+\frac{\nu^2}{\,4\kappa^2}\,\mathbbm{1}\Big)\psi\;&=\;\Big(H^{(\nu)}+\frac{\nu^2}{\,4\kappa^2}\,\mathbbm{1}\Big)\varphi_\kappa\,,
   \end{split}
  \end{equation}
  the decomposition of each $\psi$ being unique.
\end{itemize}
\end{theorem}

We observe that \eqref{eq:DomainHalpha} provides the typical decomposition of a generic element in $\mathcal{D}(H_\alpha^{(\nu)})$ into the `regular' part $\varphi_\kappa \in H^2(\mathbb{R}^3)$ and the `singular' part $\mathfrak{g}_{\nu,\kappa}\sim |x|^{-1}$ as $x\to 0$ with a precise `boundary condition' among the two.

The uniqueness property of part (ii) above is another feature that, as we shall see, emerges naturally within the Kre{\u\i}n-Vi\v{s}ik-Birman scheme. It gives the standard hydrogenoid Hamiltonian a somewhat physically distinguished status, in complete analogy with its semi-relativistic counterpart, the well-known distinguished realisation of the Dirac-Coulomb Hamiltonian (see, e.g., \cite{Gallone-AQM2017,MG_DiracCoulomb2017} and the references therein).

Last, we address the spectral analysis of each realisation $H_\alpha^{(\nu)}$. 

Since the  $H_\alpha^{(\nu)}$'s are rank-one perturbations, in the resolvent sense, of $H_{\alpha=\infty}^{(\nu)}\equiv H^{(\nu)}$,
then we deduce from \eqref{eq:spectrum_hydrogeonid} that 
\begin{equation}
 \sigma_{\mathrm{ess}}(H_\alpha^{(\nu)})\;=\;\sigma_{\mathrm{ac}}(H_\alpha^{(\nu)})\;=\;[0,+\infty)\,,\quad\sigma_{\mathrm{sc}}(H_\alpha^{(\nu)})\;=\;\varnothing\,,
\end{equation}
and only $\sigma_{\mathrm{point}}(H_\alpha^{(\nu)})$ differs from the corresponding $\sigma_{\mathrm{point}}(H^{(\nu)})$.

Concerning the corrections to $\sigma_{\mathrm{point}}(H^{(\nu)})$ due to the central perturbation, we distinguish among the two possible cases. If $\nu<0$, then the $n$-th eigenvalue $-\frac{\nu^2}{4n^2}$ in $\sigma_{\mathrm{point}}(H^{(\nu)})$ is $n^2$-fold degenerate, with partial $(2\ell+1)$-fold degeneracy in the sector of angular symmetry $\ell$ for all $\ell\in\{0,\dots,n-1\}$. All the eigenstates of $H^{(\nu)}$ with eigenvalue $-\frac{\nu^2}{4n^2}$ and with symmetry $\ell\geqslant 1$ are also eigenstates of any other realisation $H_\alpha^{(\nu)}$ with the same eigenvalue, because $H_\alpha^{(\nu)}$ is a perturbation of $H^{(\nu)}$ in the $s$-wave only. Thus, the effect of the central perturbation is a correction to the $\ell=0$ point spectrum of $H^{(\nu)}$, which consists of countably many non-degenerate eigenvalues $E_n:=-\frac{\nu^2}{4n^2}$, $n\in\mathbb{N}$.

If instead $\nu>0$, then a standard application of the Kato-Agmon-Simon Theorem (see e.g. \cite[Theorem XIII.58]{rs4}) gives $\sigma_{\mathrm{point}}(H^{(\nu)})=\varnothing$. Yet, if the central perturbation corresponds to an interaction that is attractive or at least not too much repulsive, then it can create one negative eigenvalue in the $\ell=0$ sector.

This is described in detail as follows.

\begin{theorem}[Eigenvalue corrections]
\label{thm:EV_corrections}~

\noindent For given $\alpha\in\mathbb{R}\cup\{\infty\}$ and $\nu\in\mathbb{R}$, let $\sigma^{(0)}_{\mathrm{p}}(H^{(\nu)}_\alpha)$ be point spectrum of the self-adjoint extension $H^{(\nu)}_\alpha$ with definite angular symmetry $\ell=0$ (`$s$-wave point spectrum'). Moreover, for $E<0$ let
 \begin{equation}\label{eq:Feigenvalues}
 \mathfrak{F}_\nu(E)\;:=\;\frac{\nu}{4\pi}\Big(\psi\big(1+{\textstyle\frac{\nu}{2\sqrt{|E|}}}\big)+ \ln(2 \sqrt{|E|}) +2\gamma - 1 - {\textstyle\frac{\sqrt{|E|}}{\nu}} \Big)\,.
\end{equation}
 \begin{itemize}
  \item[(i)] If $\nu<0$, then the equation
  \begin{equation}\label{eq:theoremFEalpha}
   \mathfrak{F}_\nu(E)\;=\;\alpha
  \end{equation}
 admits countably many simple negative roots that form an increasing sequence $(E_n^{(\nu,\alpha)})_{n\in\mathbb{N}}$ accumulating at zero, and 
  \begin{equation}
  \sigma^{(0)}_{\mathrm{p}}(H^{(\nu)}_\alpha)\;=\;\big\{ E_n^{(\nu,\alpha=\infty)}\,\big|\, n\in\mathbb{N}\big\}\,.
 \end{equation}
 For the Friedrichs extension,
  \begin{equation}
   E_n^{(\nu,\alpha=\infty)}\;=\;E_n^{(\nu)}\;=\;-\frac{\nu^2}{4 n^2}\,,
  \end{equation}
 that is, the ordinary hydrogenoid eigenvalues.
 \item[(ii)] If $\nu>0$, then the equation \eqref{eq:theoremFEalpha} has no negative roots if $\alpha\geqslant\alpha_\nu$, where
 \begin{equation}
  \alpha_\nu\;:=\;\frac{\nu}{4\pi}\,(\ln\nu+2\gamma-1)\,,
 \end{equation}
 and has one simple negative root $E_+^{(\nu,\alpha)}$ if $\alpha<\alpha_\nu$. Correspondingly,
 \begin{equation}
  \sigma^{(0)}_{\mathrm{p}}(H^{(\nu)}_\alpha)\;=\;
  \begin{cases}
   \varnothing & \textrm{ if }\;\alpha\geqslant\alpha_\nu\,, \\
   E_+^{(\nu,\alpha)}& \textrm{ if }\;\alpha<\alpha_\nu\,.
  \end{cases}
 \end{equation}
 \end{itemize}
\end{theorem}

\begin{figure}
\includegraphics[width=6.2cm]{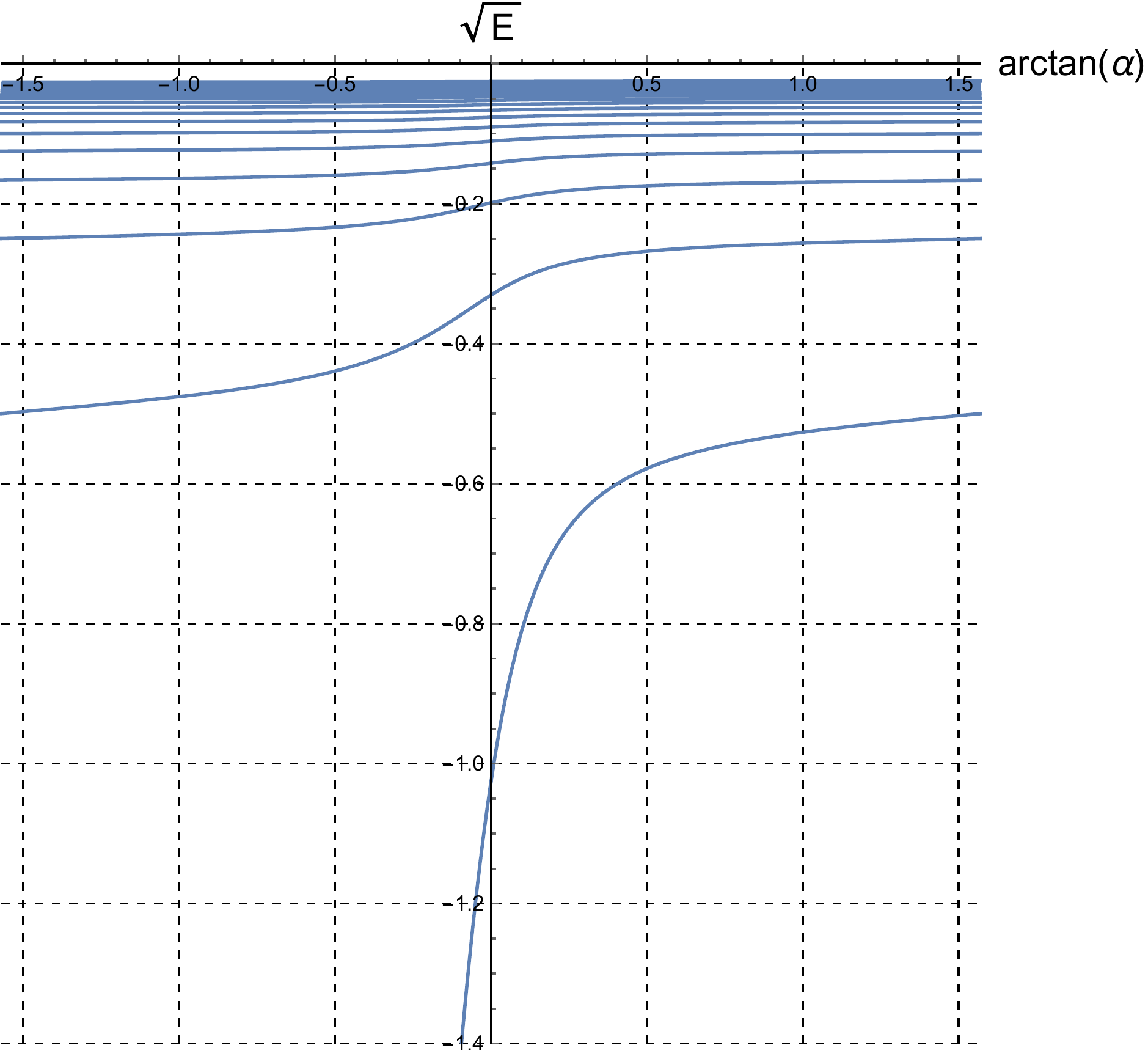}\quad
\includegraphics[width=5.8cm]{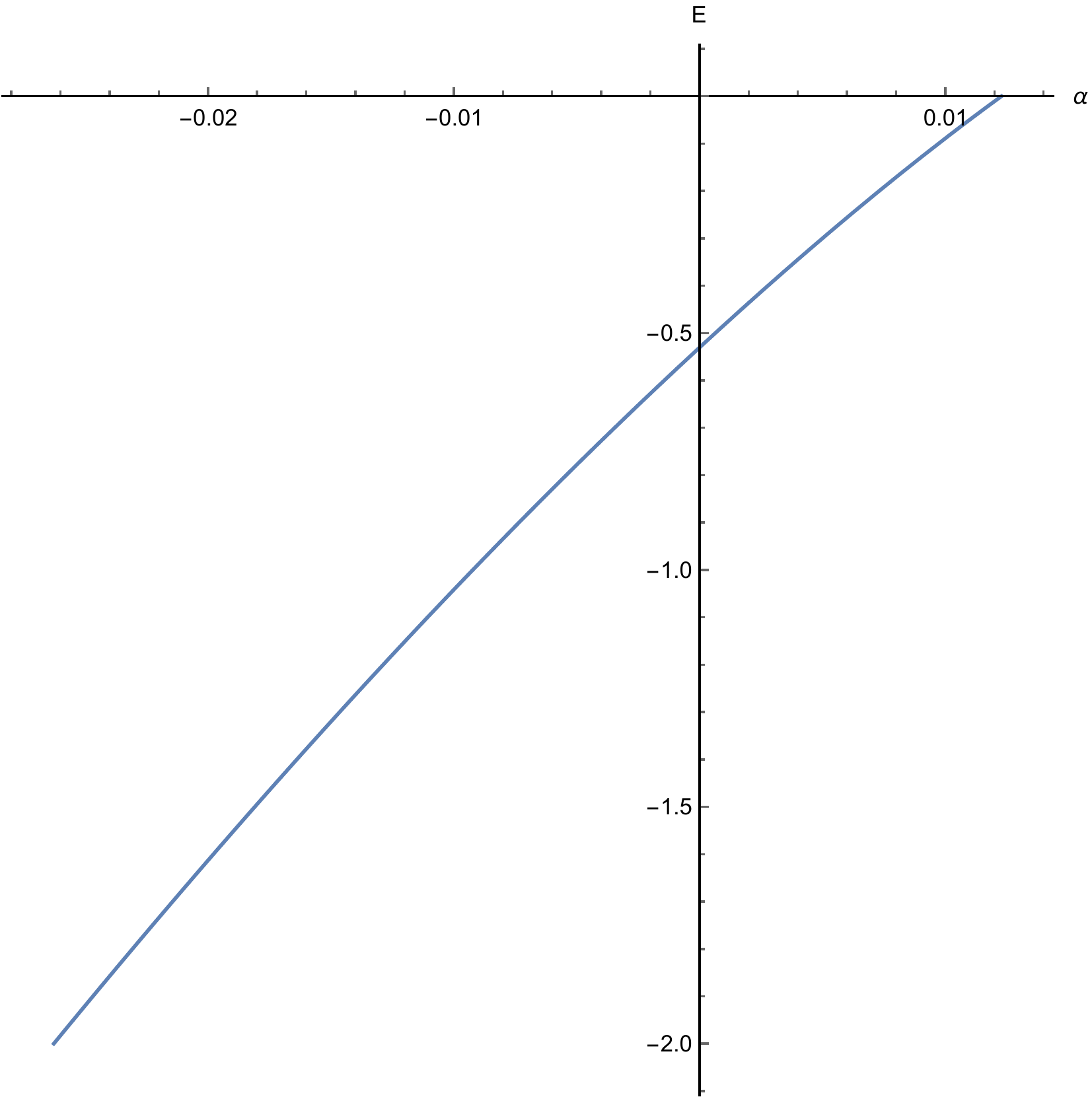}
\caption{Eigenvalues of the perturbed hydrogenoid Hamiltonian $H^{(\nu)}_\alpha$ for $\nu=-1$ (left) and $\nu=1$ (right). The scales of the energy $E$ and of the extension parameter $\alpha$ are modified to magnify the behaviour of the eigenvalues.} \label{fig:Eig}
\end{figure}

Figure  \ref{fig:Eig} displays the structure of the discrete spectrum described in Theorem \ref{thm:EV_corrections} above.

As we shall argue rigorously in due time, Figure  \ref{fig:Eig} confirms that when $ \nu < 0$ each $E_n^{(\nu,\alpha)}$ is smooth and strictly monotone in $\alpha$, with a typical \emph{fibred structure} of the union of all the discrete spectra $\sigma_{\text{disc}}(H^{(\nu)}_\alpha)$ 
\begin{equation}
(-\infty,0)\;=\bigcup_{\alpha \in (-\infty,+\infty]} \{ E_n^{(\nu,\alpha)}\, | \, n \in \mathbb{N}\} \;=\;\mathbb{R}\setminus \sigma_{\mathrm{ess}}(H_\alpha^{(\nu)})\,,
\end{equation}
(see Remark \ref{rem:spectra_fibre} below, and \cite{GM-2017-DC-EV} for an analogous phenomenon for Dirac operators),
and the correction $E_n^{(\nu,\alpha)}$ to the non-relativistic $E_n^{(\nu)}$ always \emph{decreases} the energy, with the intertwined relation $E_{n+1}^{(\nu,\alpha)}\geqslant E_n^{(\nu)}\geqslant E_n^{(\nu,\alpha)}$ (see Remark \ref{rem:EVdecreased}).


Analogously, when $\nu>0$,
\begin{equation}
(-\infty,0)\;=\bigcup_{\alpha \in (-\infty,\alpha_\nu)} \{ E_+^{(\nu,\alpha)}\} \;=\;\mathbb{R}\setminus \sigma_{\mathrm{ess}}(H_\alpha^{(\nu)})\,.
\end{equation}

\section{Self-adjoint realisations and classification}\label{sec:Section_of_Classification}

In this Section we establish the constructions of Theorems \ref{thm:1Fried}, \ref{thm:1Dclass}, and \ref{thm:3Dclass}. The main focus are the self-adjoint extensions on $L^2(\mathbb{R}^+)$ of the radial operator  $h_{0}^{(\nu)}$. Equivalently, we study the self-adjoint extensions of the shifted operator
\begin{equation}\label{eq:Symmetric}
 S\; :=\;-\frac{\ud^2}{\ud r^2} +\frac{\nu}{r}+\frac{\,\nu^2}{\,4\kappa^2}\,,\qquad\mathcal{D}(S)\;:=\;C^\infty_0(\mathbb{R}^+)\,,
\end{equation}
for generic 
\begin{equation}\label{eq:conditions_kappa}
 \kappa\,\in\,\mathbb{R}\,,\qquad\mathrm{sign}\,\kappa\;=\;-\mathrm{sign}\,\nu\,,\qquad0<|\kappa|<\textstyle\frac{1}{2}\,.
\end{equation}
Owing to \eqref{eq:spectrum_hydrogeonid}, $-\frac{\ud^2}{\ud r^2} +\frac{\nu}{r}\geqslant-\nu^2/4$, whence $S\geqslant{\textstyle\frac{1}{4}}\,\nu^2(\kappa^{-2}-1)$: thus, $S$ is densely defined and symmetric on $L^2(\mathbb{R}^+)$ with \emph{strictly positive bottom}. This feature will simplify the identification of the self-adjoint extensions of $S$: the corresponding extensions for $h_{0}^{(\nu)}$ are then obtained through a trivial shift.

It will be also convenient to make use of the notation
\begin{equation}
 \widetilde{S}\;: =\;-\frac{\ud^2}{\ud r^2} +\frac{\nu}{r}+\frac{\,\nu^2}{\,4\kappa^2} 
\end{equation}
to refer to the \emph{differential} action on functions in $L^2(\mathbb{R}^+)$, in the classical or the weak sense, with no reference to the operator domain.

In order to apply the Kre{\u\i}n-Vi\v{s}ik-Birman extension scheme \cite[Sec.~3]{GMO-KVB2017}, an amount of preparatory steps are needed (Subsect.~\ref{sec:homo} through \ref{sec:SclosureSstar}), in which we identify the spaces $\mathcal{D}(\overline{S})$, $\ker S^*$, and $S_F^{-1}\ker S^*$, $S_F$ being the Friedrichs extension of $S$. 
In Subsect.~\ref{sec:SclosureSstar} we qualify $S_F$ and prove Theorem \ref{thm:1Fried}; in Subsect.~\ref{sec:KVB_for_S} we classify the extensions of $S$ and prove Theorem \ref{thm:1Dclass}; last, in Subsect.~\ref{sec:3Dreconstr} we deduce Theorem~\ref{thm:3Dclass} from the previous results.

\subsection{The homogeneous radial problem}\label{sec:homo}~

We first qualify the space $\ker S^*$. By standard arguments (see, e.g., \cite[Lemma 15.1]{schmu_unbdd_sa}) 
\begin{equation}
 \begin{split}\label{eq:Sstar_maximal}
  \mathcal{D}(S^*)\;&=\;\left\{g\in L^2(\mathbb{R}^+)\left|\, \widetilde{S}g\in L^2(\mathbb{R}^+)\!\right.\right\} \\
  S^*g\;&=\;\widetilde{S}g\;=\;-g''+\frac{\nu}{r}\,g+\frac{\nu^2}{4\kappa^2}\,g\,,
 \end{split}
\end{equation}
that is, $S^*$ is the \emph{maximal} realisation of $\widetilde{S}$, and in fact $\overline{S}$ is the \emph{minimal} one. Thus, $\ker S^*$ is formed by the square-integrable solutions to $\widetilde{S} u = 0$ on $\mathbb{R}^+$. It is also standard (see e.g. \cite[Theorems 5.2--5.4]{Wasow_asympt_expansions}) that if $u$ solves $\widetilde{S} u = 0$, then it is smooth on $\mathbb{R}^+$, with possible singularity only at zero or infinity.

Through the change of variable $\rho:=-\frac{\nu}{\kappa}r$, $w(\rho):=u(r)$, where $-\frac{\nu}{\kappa}>0$ for every non-zero $\nu$ owing to \eqref{eq:conditions_kappa}, the differential problem becomes
\begin{equation}\label{eq:hypergeom}
 \Big(-\frac{\ud^2}{\ud \rho^2} -\frac{\kappa}{\rho}+\frac{1}{4}\Big)w\;=\;0\,,
\end{equation}
that is, a special case of Whittaker's equation $w''-(\frac{1}{4}-\frac{\kappa}{\rho}+(\frac{1}{4}-\mu^2)\frac{1}{\rho^2})w=0$ with parameter $\mu=\frac{1}{2}$ \cite[Eq.~(13.1.31)]{Abramowitz-Stegun-1964}.
The functions
\begin{eqnarray}
 \mathscr{M}_{\kappa,\frac{1}{2}}(\rho) &=& e^{-\frac{1}{2} \rho} \rho\, M_{1-\kappa,2}(\rho) \\
 \mathscr{W}_{\kappa,\frac{1}{2}}(\rho) &=& e^{-\frac{1}{2} \rho} \rho \,U_{1-\kappa,2}(\rho)
\end{eqnarray}
form a pair $(\mathscr{M}_{\kappa,\frac{1}{2}},\mathscr{W}_{\kappa,\frac{1}{2}})$ of linearly independent solutions to \eqref{eq:hypergeom} \cite[Eq.~(13.1.32)-(13.1.33)]{Abramowitz-Stegun-1964}, where $M_{a,b}$ and $U_{a,b}$ are, respectively, Kummer's and Tricomi's function \cite[Eq.~(13.1.2)-(13.1.3)]{Abramowitz-Stegun-1964}.

Owing to \cite[Eq.~(13.5.5), (13.5.7), (13.1.2) and (13.1.6)]{Abramowitz-Stegun-1964} as $\rho\to 0$, and to \cite[Eq.~(13.1.4) and (13.1.8)]{Abramowitz-Stegun-1964} as $\rho\to +\infty$, one has the asymptotics
\begin{equation}\label{eq:AsymM0}
 \begin{split}
 \mathscr{M}_{\kappa,\frac{1}{2}}(\rho)\;&\stackrel{\rho\to 0}{=}\; \rho -\frac{\kappa}{2} \rho^2 + \frac{1+2\kappa^2}{24} \rho^3 + O(\rho^4) \\
 \mathscr{W}_{\kappa,\frac{1}{2}}(\rho)\;&\stackrel{\rho\to 0}{=}\;\frac{1}{\Gamma(1-\kappa)} -\frac{ \kappa}{\Gamma(1-\kappa)} \,\rho \ln \rho \\
 &\qquad\quad +\frac{\,(2-4 \gamma) \kappa-2 \kappa\psi(1-\kappa)  -1\,}{2 \Gamma(1-\kappa)}\, \rho+O(\rho^2 \ln \rho)
 \end{split}
\end{equation}
and 
\begin{equation}\label{eq:AsymMInf}
\begin{split}
 \mathscr{M}_{\kappa,\frac{1}{2}}(\rho)\;&\stackrel{\rho\to +\infty}{=}\;\frac{1}{\Gamma(1-\kappa)} e^{\rho/2} \rho^{-\kappa} (1+O(\rho^{-1})) \\
 \mathscr{W}_{\kappa,\frac{1}{2}}(\rho)\;&\stackrel{\rho\to +\infty}{=}\;e^{-\rho/2} \rho^\kappa (1+O(\rho^{-1}))\,,
\end{split}
\end{equation}
where $\gamma\sim 0.577$ is the Euler-Mascheroni constant and $\psi(z)=\Gamma'(z)/\Gamma(z)$ is the digamma function. Since $0<|\kappa|<\frac{1}{2}$, the expressions \eqref{eq:AsymM0} and \eqref{eq:AsymMInf} make sense.

Therefore only $\mathscr{W}_{\kappa,\frac{1}{2}}$ is square-integrable at infinity, whereas both $\mathscr{M}_{\kappa,\frac{1}{2}}$ and $\mathscr{W}_{\kappa,\frac{1}{2}}$ are square-integrable at zero. This implies that the square-integrable solutions to $\widetilde{S}u \; = \;0$ form a \emph{one}-dimensional space, that is, $\dim\ker S^*=1$. 

Explicitly, upon setting
\begin{equation}\label{eq:FPhi}
 \begin{split}
  F_\kappa(r)\;&:=\;\mathscr{M}_{\kappa,\frac{1}{2}}(\lambda r) \\
  \Phi_\kappa(r)&:=\;\mathscr{W}_{\kappa,\frac{1}{2}}(\lambda r)\,,\qquad \lambda:=-{\textstyle\frac{\nu}{\kappa}}>0\,,
 \end{split}
\end{equation}
one has that
\begin{equation}\label{eq:kerSstar}
 \ker S^*\;=\;\mathrm{span}\{\Phi_\kappa\}
\end{equation}
and that $(F_\kappa,\Phi_\kappa)$ is a pair of linearly independent solutions to the original problem $\widetilde{S}u=0$.

\subsection{Inhomogeneous inverse radial problem}~

Next, let us focus on the inhomogeneous problem $\widetilde{S}f=g$ in the unknown $f$ for given $g$. With respect to the fundamental system $(F_\kappa,\Phi_\kappa)$ for $\widetilde{S}u=0$, the general solution is given by
\begin{equation}
 f\;=\;c_1 F_\kappa + c_2 \Phi_\kappa + f_{\mathrm{part}}
\end{equation}
for $c_1,c_2\in\mathbb{C}$ and some particular solution $f_{\mathrm{part}}$, i.e., $\widetilde{S}f_{\mathrm{part}}=g$.

The Wronskian
\begin{equation}
W(\Phi_\kappa,F_\kappa)(r)\;:=\;\det \begin{pmatrix}
\Phi_\kappa(r) & F_\kappa(r) \\
\Phi_\kappa'(r) & F_\kappa'(r)
\end{pmatrix}
\end{equation}
relative to the pair $(F_\kappa,\Phi_\kappa)$ is actually constant in $r$, owing to Liouville's theorem, with a value that can be computed by means of the asymptotics \eqref{eq:AsymM0} or \eqref{eq:AsymMInf} and amounts to
\begin{equation}\label{eq:value_of_W}
W(\Phi_\kappa,F_\kappa)\;=\;\frac{-\nu/\kappa}{\Gamma(1-\kappa)}\;=:\; W\,.
\end{equation}

A standard application of the method of variation of constants \cite[Section 2.4]{Wasow_asympt_expansions} shows that we can take $f_{\mathrm{part}}$ to be
\begin{equation}
f_{\text{part}}(r) \;=\; \int_0^{+\infty} G(r,\rho) g(\rho) \, \ud \rho\,,
\end{equation}
where 
\begin{equation}\label{eq:Green}
G(r,\rho) \, := \,\frac{1}{W} \begin{cases}
\Phi_\kappa(r) F_\kappa(\rho) \qquad \text{if }0 < \rho < r\\
F_\kappa(r) \Phi_\kappa(\rho) \qquad \text{if } 0 < r < \rho \,.
\end{cases}
\end{equation}

The following property holds.

\begin{lemma}\label{lem:RGbddsa}
The integral operator $R_G$ on $L^2(\mathbb{R}^+,\ud r)$ with kernel $G(r,\rho)$ given by \eqref{eq:Green} is bounded and self-adjoint.
\end{lemma}

\begin{proof}
$R_G$ splits into the sum of four integral operators with kernels given by
\begin{equation*}
 \begin{split}
  G^{++}(r,\rho)\;&:=\;G(r,\rho)\,\mathbf{1}_{(1,+\infty)}(r)\,\mathbf{1}_{(1,+\infty)}(\rho) \\
  G^{+-}(r,\rho)\;&:=\;G(r,\rho)\,\mathbf{1}_{(1,+\infty)}(r)\,\mathbf{1}_{(0,1)}(\rho) \\
  G^{-+}(r,\rho)\;&:=\;G(r,\rho)\,\mathbf{1}_{(0,1)}(r)\,\mathbf{1}_{(1,+\infty)}(\rho) \\
  G^{--}(r,\rho)\;&:=\;G(r,\rho)\,\mathbf{1}_{(0,1)}(r)\,\mathbf{1}_{(0,1)}(\rho)\,, 
 \end{split}
 \end{equation*}
where $\mathbf{1}_J$ denotes the characteristic function of the interval $J\subset\mathbb{R}^+$.
We can estimate each $G^{LM}(r,\rho)$, $L,M\in\{+,-\}$, by means of the short and large distance asymptotics \eqref{eq:AsymM0}-\eqref{eq:AsymMInf} for $F_\kappa$ and $\Phi_\kappa$. Calling $\lambda=-\frac{\nu}{\kappa}$ as in \eqref{eq:FPhi}, for example, 
\[
 \begin{split}
  |\Phi_\kappa(r) \,F_\kappa(\rho)\,\mathbf{1}_{(1,+\infty)}(r)\,\mathbf{1}_{(1,+\infty)}(\rho)|\;&\lesssim\; e^{-\frac{\lambda}{2}(r-\rho)} \,\left(\frac{r}{\rho}\right)^\kappa \qquad\textrm{if }\,0<\rho<r \\
  |F_\kappa(r)\, \Phi_\kappa(\rho)\,\mathbf{1}_{(1,+\infty)}(r)\,\mathbf{1}_{(1,+\infty)}(\rho)|\;&\lesssim\; e^{-\frac{\lambda}{2}(\rho-r)} \, \left(\frac{\rho}{r}\right)^\kappa \qquad\;\textrm{if }\,0<r<\rho\,,
 \end{split}
\]
because $F_k$ diverges exponentially and $\Phi_\kappa$ vanishes exponentially as $r\to+\infty$. Thus,
\[
 |G^{++}(r,\rho)|\;\lesssim\;e^{-\frac{\lambda}{4}|r-\rho|}\,.
\]
With analogous reasoning we find
\begin{equation*}\tag{*}\label{eq:matrix_estimates}
 \begin{split}
  |G^{++}(r,\rho)|\;&\lesssim\;e^{-\frac{\lambda}{4}|r-\rho|}\,\mathbf{1}_{(1,+\infty)}(r)\,\mathbf{1}_{(1,+\infty)}(\rho) \\
  |G^{+-}(r,\rho)|\;&\lesssim\;e^{-\frac{\lambda}{4}r}\,\mathbf{1}_{(1,+\infty)}(r)\,\mathbf{1}_{(0,1)}(\rho) \\
  |G^{-+}(r,\rho)|\;&\lesssim\;e^{-\frac{\lambda}{4}\rho}\,\mathbf{1}_{(0,1)}(r)\,\mathbf{1}_{(1,+\infty)}(\rho) \\
  |G^{--}(r,\rho)|\;&\lesssim\;\mathbf{1}_{(0,1)}(r)\,\mathbf{1}_{(0,1)}(\rho)\,.
 \end{split}
\end{equation*}
The last three bounds in \eqref{eq:matrix_estimates} imply $G^{+-},G^{-+},G^{--}\in L^2(\mathbb{R}^+\times\mathbb{R}^+,\ud r\,\ud\rho)$ and therefore the corresponding integral operators are Hilbert-Schmidt operators on $L^2(\mathbb{R}^+)$. The first bound in \eqref{eq:matrix_estimates} allows to conclude, by an obvious Schur test, that also the integral operator with kernel $G^{++}(r,\rho)$ is bounded on $L^2(\mathbb{R}^+)$. This proves the overall boundedness of $R_G$. Its self-adjointness is then clear from \eqref{eq:Green}: the adjoint $R_G^*$ of $R_G$ has kernel $\overline{G(\rho,r)}$, but $G$ is real-valued and $G(\rho,r)=G(r,\rho)$, thus proving that $R_G^*=R_G$.
\end{proof}

\subsection{Distinguished extension and its inverse}~


In the Kre\u{\i}n-Vi\v{s}ik-Birman scheme one needs a \emph{reference} self-adjoint extension of $S$ with everywhere defined bounded inverse: the Friedrichs extension $S_F$ is surely so, since the bottom of $S$ is strictly positive by construction.

In this Subsection we shall prove the following.
\begin{proposition}\label{eq:RGisSFinv}
 $R_G=S_F^{-1}$.
\end{proposition}

This is checked in several steps. First, we recognise that $R_G$ inverts a self-adjoint extension of $S$.

\begin{lemma}\label{eq:RGinvertsExtS}
 There exists a self-adjoint extension $\mathscr{S}$ of $S$ in $L^2(\mathbb{R}^+)$ which has everywhere defined and bounded inverse and such that $\mathscr{S}^{-1}=R_G$.
\end{lemma}

\begin{proof}
 $R_G$ is bounded and self-adjoint (Lemma \ref{lem:RGbddsa}), and by construction satisfies $\widetilde{S}\,R_G\,g=g$ $\forall g\in L^2(\mathbb{R}^+)$. Therefore, $R_G g=0$ for some $g\in L^2(\mathbb{R}^+)$ implies $g=0$, i.e., $R_G$ is injective. Then $R_G$ has dense range ($(\mathrm{ran}\,R_G)^\perp=\ker R_G$). As such (see, e.g., \cite[Theorem 1.8(iv)]{schmu_unbdd_sa}), $\mathscr{S}:=R_G^{-1}$ is self-adjoint. One thus has $R_G=\mathscr{S}^{-1}$ and from the identity $S^*R_G=\mathbbm{1}$ on $L^2(\mathbb{R}^+)$ one deduces that for any $f\in\mathcal{D}(\mathscr{S})$, say, $f=R_G g=\mathscr{S}^{-1} g$ for some $g\in L^2(\mathbb{R}^+)$, the identity $S^*f=\mathscr{S}f$ holds. This means that $S^*\supset\mathscr{S}$, whence also $\overline{S}=S^{**}\subset\mathscr{S}$, i.e., $\mathscr{S}$ is a self-adjoint extension of $S$.
\end{proof}

Next, we recall the following concerning the form of the Friedrichs extension. Let us define
 \begin{equation}\label{eq:NormF}
  \|f\|_F^2\;:=\;\langle f,Sf\rangle+\langle f,f\rangle\,,
 \end{equation}
 which, for $f\in C^\infty_0(\mathbb{R}^+)$, is a norm equal to
 \begin{equation}\label{eq:NormF2}
  \|f\|_F^2\;=\;\|f'\|_{L^2}^2+\nu\|r^{-\frac{1}{2}}f\|_{L^2}^2+(\textstyle{\frac{\nu^2}{\,4\,\kappa^2}}+1)\|f\|_{L^2}^2\,.
 \end{equation}

\begin{lemma}\label{lem:Fform}
 The quadratic form of the Friedrichs extension of $S$ is given by
 \begin{equation}\label{eq:Fform}
  \begin{split}
  \mathcal{D}[S_F]\;&=\;\big\{f\in L^2(\mathbb{R}^+)\,\big|\, \|f'\|_{L^2}^2+\nu\,\|r^{-\frac{1}{2}}f\|_{L^2}^2+\|f\|_{L^2}^2<+\infty\big\} \\
   S_F[f,h]\;&=\;\int_0^{+\infty}\!\!\Big(\,\overline{f'(r)}h'(r)+\nu\,\frac{\,\overline{f(r)}h(r)}{r}+\frac{\nu^2}{\,4\,\kappa^2}\,\overline{f(r)}h(r)\Big)\ud r\,.
  \end{split}
 \end{equation}
 \end{lemma}

\begin{proof}
 A standard construction (see, e.g., \cite[Theorem A.2]{GMO-KVB2017}), that follows from the fact that $\mathcal{D}[S_F]$ is the closure of $\mathcal{D}(S)=C^\infty_0(\mathbb{R}^+)$ in the norm $\|\cdot\|_F$: then \eqref{eq:Fform} follows at once from \eqref{eq:NormF}-\eqref{eq:NormF2}.
\end{proof}

In fact, the Friedrichs form domain is a classical functional space.

\begin{lemma}\label{lem:DomSFrid}
$\mathcal{D}[S_F]\:=\:H^1_0(\mathbb{R}^+)\::=\:\overline{C^\infty_0(\mathbb{R}^+)}^{\Vert \, \Vert_{H^1}}$.
\end{lemma}

\begin{proof}
Hardy's inequality
\begin{equation*}
\Vert r^{-1} f \Vert_{L^2} \;\leqslant\; 2 \,\Vert f' \Vert_{L^2}\qquad \forall\,f \in C^\infty_0(\mathbb{R}^+)
\end{equation*}
implies
\begin{equation*}
\Vert r^{-\frac{1}{2}} f \Vert^2_{L^2} \;\leqslant\; \frac{\varepsilon}{4} \,\Vert r^{-1} f \Vert^2_{L^2} +\frac{1}{\varepsilon}\, \Vert f \Vert^2_{L^2} \;\leqslant\; \varepsilon \Vert f' \Vert_{L^2}^2 + \varepsilon^{-1} \Vert f \Vert_{L^2}^2
\end{equation*}
for arbitrary $\varepsilon>0$. This and \eqref{eq:NormF2} imply on the one hand $\|f\|_F\lesssim\|f\|_{H^1}$, and on the other hand
\[
 \|f\|_F^2\;\geqslant\;(1-|\nu|\varepsilon)\|f'\|_{L^2}^2+(1+\textstyle{\frac{\nu^2}{\,4\,\kappa^2}}-|\nu|\varepsilon^{-1})\|f\|_{L^2}^2\,.
\]
The r.h.s.~above is equivalent to the $H^1$-norm provided that the coefficients of $\|f'\|_{L^2}^2$ and $\|f\|_{L^2}^2$ are strictly positive, which is the same as
\[
 \nu^2\;<\;\frac{|\nu|}{\varepsilon}\;<\;1+\frac{\nu^2}{\,4\,\kappa^2}\,.
\]
For given $\nu$ and $\kappa$, a choice of $\varepsilon>0$ satisfying the inequalities above is always possible, because $\nu^2<1+\frac{\nu^2}{\,4\,\kappa^2}$, or equivalently, $1+\nu^2(\frac{1}{4\kappa^2}-1)>0$, which is true owing to the assumption $0<|\kappa|<\frac{1}{2}$. We have therefore shown that $\|f\|_F\approx\|f\|_{H^1}$ in the sense of the equivalence of norms on $C^\infty_0(\mathbb{R}^+)$. Now, the $\|\cdot\|_F$-completion of $C^\infty_0(\mathbb{R}^+)$ is by definition $\mathcal{D}[S_F]$, whereas the $\|\cdot\|_{H^1}$-completion is $H^1_0(\mathbb{R}^+)$: the Lemma is therefore proved.
\end{proof}

Let us now highlight the following feature of $\mathrm{ran} \,R_G$.

\begin{lemma}\label{lem:ranRGinDrinv}
For every $g \in L^2(\mathbb{R}^+)$ one has
\begin{equation}\label{eq:Integral}
\int_0^{+\infty} \frac{| (R_G g)(r) |^2}{r^2} \ud r \;<\; +\infty\,,
\end{equation}
i.e.,
\begin{equation}
\mathrm{ran} \, R_G \subset \mathcal{D}(r^{-1})\,.
\end{equation}
\end{lemma}

\begin{proof}
It suffices to prove the finiteness of the integral in \eqref{eq:Integral} only for $r \in (0,1)$, since $\int_1^{+\infty} r^{-2} | (R_Gg)(r) |^2 \,\ud r \leqslant \Vert R_G \Vert^2 \Vert g \Vert^2_{L^2}$. Owing to \eqref{eq:value_of_W} and \eqref{eq:Green},
\begin{equation*}\tag{*}
|(R_Gg)(r)|\;\lesssim\; |\Phi_\kappa(r)| \int_0^r |F_\kappa(\rho) g(\rho) | \, \ud \rho + |F_\kappa(r)| \int_0^{+\infty} |\Phi_\kappa(\rho) g(\rho)| \,\ud \rho\,.
\end{equation*}
We then exploit the asymptotics \eqref{eq:AsymM0}.
The first summand in the r.h.s.~above as a $O(r^{3/2})$-quantity as $r\downarrow 0$, because in this limit  $\Phi_\kappa$ is smooth and bounded, whereas $F_\kappa$ is smooth and vanishes as $O(r)$, and therefore
\begin{equation*}
\int_0^r |F_\kappa(\rho) g(\rho)| \,\ud \rho \;\leqslant\sup_{\rho \in [0,r]}|F_\kappa(\rho)|\, \Vert g \Vert_{L^2} \,r^{1/2}\,=\, O(r^{3/2}) \,.
\end{equation*} 
The second  summand in the r.h.s.~of (*) is a $O(r)$-quantity as $r\downarrow 0$, because so is $F_\kappa(r)$ and because $\int_0^{+\infty}|\Phi_\kappa(\rho) g(\rho)|\, \ud \rho \leqslant \Vert \Phi_\kappa \Vert_{L^2} \Vert g \Vert_{L^2}$. Thus, $(R_Gg)(r)=O(r)$ as $r\downarrow 0$, whence the integrability of $r^{-2}|(R_Gg)(r)|^2$ at zero.
\end{proof}

We can finally prove that $R_G=S_F^{-1}$.

\begin{proof}[Proof of Proposition \ref{eq:RGisSFinv}]
 $R_G=\mathscr{S}^{-1}$ for some $\mathscr{S}=\mathscr{S}^*\supset S$ (Lemma \ref{eq:RGinvertsExtS}), and we want to conclude that $\mathscr{S}=S_F$. This follows if we show that $\mathcal{D}(\mathscr{S})\subset \mathcal{D}[S_F]$, owing to the well-known property of $S_F$ that distinguishes it from all other self-adjoint extensions of $S$.
 
 Let us then pick a generic $f=R_G g\in\mathrm{ran}\,R_G=\mathcal{D}(\mathscr{S})$ for some $g\in L^2(\mathbb{R}^+)$ and show that $S_F[f]:=S_F[f,f]<+\infty$, the form of $S_F$ being given by Lemma \ref{lem:Fform}. The fact that $\|f\|_{L^2}^2$ is finite is obvious, and the finiteness of $\|r^{-\frac{1}{2}}f\|_{L^2}^2$ follows by interpolation from Lemma \ref{lem:ranRGinDrinv}. We are thus left with proving that $\|f'\|_{L^2}^2<+\infty$, and the conclusion then follows from \eqref{eq:Fform}.
 
 Now, $f\in\mathcal{D}(S^*)$ and therefore $-f''+\frac{\nu}{r} f+\frac{\nu^2}{\,4\,\kappa^2} f=g\in L^2(\mathbb{R}^+)$: this, and the already mentioned square-integrability of $f$ and $r^{-1}f$, yield $f''\in L^2(\mathbb{R}^+)$. It is then standard (see, e.g., \cite[Remark 4.21]{Grubb-DistributionsAndOperators-2009}) to deduce that $f'$ too belongs to $L^2(\mathbb{R}^+)$, thus concluding the proof.
%
\end{proof}

For later purposes we set for convenience
\begin{equation}\label{eq:Psikappa}
 \Psi_\kappa\;:=\;S_F^{-1}\Phi_\kappa\;=\;R_G\Phi_\kappa
\end{equation}
and we prove the following.

\begin{lemma}\label{lem:RGPhi_atzero} One has
\begin{equation}\label{eq:RGPhi_atzero}
\Psi_\kappa(r) \;=\; \Gamma(1- \kappa) \Vert \Phi_\kappa \Vert_{L^2}^2\, r + O(r^2)\qquad\textrm{as } r \downarrow 0\,.
\end{equation}
\end{lemma}

\begin{proof}
Owing to \eqref{eq:value_of_W} and \eqref{eq:Green},
\begin{equation*}
(R_G \Phi_\kappa)(r) \;= \; -\frac{\kappa\,\Gamma(1-\kappa)}{\nu} \Big(\Phi_\kappa(r)\! \int_0^r F_\kappa(\rho) \Phi_\kappa(\rho) \,\ud \rho + F_\kappa(r)\! \int_r^{+\infty} \!\!\Phi_\kappa^2(\rho) \,\ud \rho\Big)\,.
\end{equation*}
As $r\downarrow 0$, \eqref{eq:AsymM0} and \eqref{eq:FPhi} imply that the first summand behaves as
\begin{equation*}
{\textstyle -\frac{\,\kappa\Gamma(1-\kappa)}{\nu}}\,\big({\textstyle\frac{1}{\Gamma(1-\kappa)}}+ O(r \ln r)\big)\! \int_0^r \!\big({\textstyle-\frac{\nu}{\kappa}} \rho+O(\rho^2)\big) \big({\textstyle\frac{1}{\Gamma(1-\kappa)}}+O(\rho\ln\rho) \big) \, \ud \rho\,,
\end{equation*}
which, after some simplifications, becomes
\begin{equation*}
\frac{1}{\,2\,\Gamma(1-\kappa)}\, r^2 + O(r^3 \ln r)\,.
\end{equation*}
The second summand turns out to be the leading term: indeed, as $r\downarrow 0$, $\int_r^\infty \Phi_\kappa^2\,\ud\rho = \| \Phi_\kappa \|_{L^2(\mathbb{R}^+)}^2 + O(r)$ and hence 
\begin{equation*}
-\frac{\kappa\,\Gamma(1-\kappa)}{\nu} \, F_\kappa(r)\! \int_r^{+\infty}\!\! \Phi_\kappa^2(\rho) \,\ud \rho\;=\;\Gamma(1- \kappa) \Vert \Phi_\kappa \Vert_{L^2}^2\, r + O(r^2)\,,
\end{equation*}
which completes the proof.
\end{proof}

\subsection{Operators $\overline{S}$, $S_F$, and $S^*$}\label{sec:SclosureSstar}~

In general (see \cite[Theorem 2.2]{GMO-KVB2017} and \cite[Eq.~(2.6)]{GMO-KVB2017}), the space $\mathcal{D}(S^*)$ implicitly qualified in \eqref{eq:Sstar_maximal} and the space $\mathcal{D}(S_F)$ have the following internal structure:
\begin{eqnarray}
 \mathcal{D}(S^*)&=&\mathcal{D}(\overline{S})\dotplus S_F^{-1}\ker S^*\dotplus\ker S^* \\
 \mathcal{D}(S_F)&=&\mathcal{D}(\overline{S})\dotplus S_F^{-1}\ker S^*\,.
\end{eqnarray}
Owing to \eqref{eq:kerSstar} and to \eqref{eq:Psikappa}, this reads
\begin{eqnarray}
 \mathcal{D}(S^*)&=&\big\{ g=f+ c_1 \Psi_\kappa +c_0\Phi_\kappa\,|\,f\in\mathcal{D}(\overline{S}),\,c_0,c_1\in\mathbb{C}\big\}\label{eq:DSstar_firstversion} \\
 \mathcal{D}(S_F)&=&\mathcal{D}(\overline{S})\dotplus \mathrm{span}\{\Psi_\kappa\}\,. \label{eq:DSF_firstversion}
\end{eqnarray}

Let us focus on the space $\mathcal{D}(\overline{S})$. As observed, e.g., in \cite[Prop.~3.1(i)-(ii)]{DR-2017}, the functions in $\mathcal{D}(\overline{S})$ display the following features.
\begin{lemma}\label{lem:Derez}
Let $f \in \mathcal{D}(\overline{S})$. Then the functions $f$ and $f'$
\begin{itemize}
 \item[(i)] are continuous on $\mathbb{R}^+$ and vanish as $r\to +\infty$;
 \item[(ii)] vanish as $r\downarrow 0$ as
 \begin{equation}\label{eq:ffprimeDSclosure}
\begin{split}
f(r) = o(r^{3/2})\,, \qquad f'(r) = o (r^{1/2})\,.
\end{split}
\end{equation}
\end{itemize}
\end{lemma}

We can then conclude the following.
\begin{lemma}\label{lem:DSclosure_is_H20}
 One has
 \begin{equation}\label{eq:DSclosure_is_H20}
  \mathcal{D}(\overline{S})\;=\;H^2_0(\mathbb{R}^+)\;=\;\overline{\,C^\infty_0(\mathbb{R}^+)\,}^{\|\,\|_{H^2}}\,.
 \end{equation}
\end{lemma}

\begin{proof}
 First we observe that
 \[\tag{i}
  \mathcal{D}(\overline{S})\;\subset\; H^2_0(\mathbb{R}^+)\,.
 \]
 Indeed, for any $f\in\mathcal{D}(\overline{S})$ one has $\widetilde{S}f=-f''+\frac{\nu}{r}f+\frac{\nu^2}{4\kappa^2}f\in L^2(\mathbb{R}^+)$, as well as $f\in L^2(\mathbb{R}^+)$ and $r^{-1}f\in L^2(\mathbb{R}^+)$, the latter following from \eqref{eq:ffprimeDSclosure}; therefore, $f''\in L^2(\mathbb{R}^+)$ and hence, as recalled already, necessarily $f\in H^2(\mathbb{R}^+)$. Owing to \eqref{eq:ffprimeDSclosure} again, $f(0)=f'(0)=0$, whence $f\in H^2_0(\mathbb{R}^+)$. 
 
 We also have the inclusion
  \[\tag{ii}
  H^2_0(\mathbb{R}^+)\;\subset\;\mathcal{D}(S^*)\,.
 \]
 Indeed, for any $f\in H^2_0(\mathbb{R}^+)$ one has $f,f''\in L^2(\mathbb{R}^+)$, and $f\in C^1_0(\mathbb{R}^+)$ by Sobolev's Lemma, where $C^1_0(\mathbb{R}^+)$ is the space of the $C^1$-functions over $\mathbb{R}^+$ vanishing at zero together with their derivative. Thus, $f(r)=o(r)$ as $r\downarrow 0$, implying $r^{-1}f\in L^2(\mathbb{R}^+)$.
 Then $\widetilde{S}f=-f''+\frac{\nu}{r}f+\frac{\nu^2}{4\kappa^2}f\in L^2(\mathbb{R}^+)$, which by \eqref{eq:Sstar_maximal} means that $f\in \mathcal{D}(S^*)$.
 
 We have then the chain
 \[
 \begin{split}
  \mathcal{D}(\overline{S})\;&\subset\; H^2_0(\mathbb{R}^+)\;\subset\;\mathcal{D}(S^*)\;=\;\mathcal{D}(\overline{S})\dotplus\mathrm{span}\{\Psi_\kappa,\Phi_\kappa\} \\
  &\subset\; H^2_0(\mathbb{R}^+)\dotplus\mathrm{span}\{\Psi_\kappa,\Phi_\kappa\}\;\subset\;\mathcal{D}(S^*)\,,
 \end{split}
 \]
 where the first two inclusions are (i) and (ii) respectively, the identity that follows is an application of \eqref{eq:DSstar_firstversion}, then the next inclusion follows from (i) again and the sum remains direct because no non-zero element in $\mathrm{span}\{\Psi_\kappa,\Phi_\kappa\}$ belongs to $H^2_0(\mathbb{R}^+)$, and the last inclusion follows from (ii) and \eqref{eq:DSstar_firstversion}. Therefore,
 \[
  \mathcal{D}(\overline{S})\dotplus\mathrm{span}\{\Psi_\kappa,\Phi_\kappa\}\;=\;H^2_0(\mathbb{R}^+)\dotplus\mathrm{span}\{\Psi_\kappa,\Phi_\kappa\}\,,\quad\textrm{with }\mathcal{D}(\overline{S})\,\subset\, H^2_0(\mathbb{R}^+)\,,
 \]
 whence necessarily $\mathcal{D}(\overline{S})=H^2_0(\mathbb{R}^+)$.
 \end{proof}
 
 As a consequence, \eqref{eq:DSF_firstversion} now reads
 \begin{equation}\label{eq:DSF2}
  \mathcal{D}(S_F)\;=\;H^2_0(\mathbb{R}^+)\dotplus\mathrm{span}\{\Psi_\kappa\}
 \end{equation}
  and in addition we can qualify $\mathcal{D}(S_F)$ as follows.
 
 \begin{lemma}\label{lem:DomOpFridS} One has
  \begin{equation}\label{eq:DSF_z}
   \begin{split}
    \mathcal{D}(S_F)\;&=\;H^2(\mathbb{R}^+)\cap H^1_0(\mathbb{R}^+) \\
    &=\;\big\{f\in H^2(\mathbb{R}^+)\,|\,f(0)=O(r)\;\textrm{as}\;r\downarrow 0\big\}\,.
   \end{split}
  \end{equation}
 \end{lemma}

  \begin{proof}
  Based on \eqref{eq:DSF_firstversion} and \eqref{eq:DSclosure_is_H20}, let $\phi=f+c\,\Phi_k\in \mathcal{D}(S_F)$ for generic $f\in H^2_0(\mathbb{R}^+)$ and $c\in\mathbb{C}$. From $-\Psi_\kappa''+\frac{\nu}{r}\Psi_k+\frac{\nu^2}{4\kappa^2}\Psi_\kappa=S_F\Psi_\kappa=\Phi_\kappa\in L^2(\mathbb{R}^+)$ and from Lemma \ref{lem:ranRGinDrinv} one deduces that $\Psi_\kappa''\in L^2(\mathbb{R}^+)$ and hence $\Psi_\kappa\in H^2(\mathbb{R}^+)$, which proves that $\mathcal{D}(S_F)\subset H^2(\mathbb{R}^+)$. Moreover, $\mathcal{D}(S_F)\subset\mathcal{D}[S_F]=H^1_0(\mathbb{R}^+)$, owing to Lemma \ref{lem:DomSFrid}, whence the conclusion $\mathcal{D}(S_F)\subset H^2(\mathbb{R}^+)\cap H^1_0(\mathbb{R}^+)$.

  For the converse inclusion, any $\phi\in H^2(\mathbb{R}^+)\cap H^1_0(\mathbb{R}^+)$ is re-written as $\phi=f+\frac{\phi'(0)}{\Psi_\kappa'(0)}\Psi_\kappa$ with $f:=\phi-\frac{\phi'(0)}{\Psi_\kappa'(0)}\Psi_\kappa$ (it is clear from the proof of Lemma \ref{lem:RGPhi_atzero} that $\Psi_\kappa'(0)=-\Gamma(1-\kappa)\|\Phi_\kappa\|_{L^2}^2\neq 0$). By linearity $f\in H^2(\mathbb{R}^+)$, by the assumptions on $\phi$ and \eqref{eq:RGPhi_atzero} $f(0)=0$, and by construction $f'(0)=0$. Thus, $f\in H^2_0(\mathbb{R}^+)$. Then $\phi\in\mathcal{D}(S_F)$ owing to \eqref{eq:DSF2}.
  \end{proof}

 In turn, we can now re-write \eqref{eq:DSstar_firstversion} as
 \begin{equation}\label{eq:DSstar_secondversion}
  \begin{split}
   \mathcal{D}(S^*)\;&=\;H^2_0(\mathbb{R}^+)\dotplus\mathrm{span}\{\Psi_\kappa,\Phi_\kappa\} \\
   &=\;\big(H^2(\mathbb{R}^+)\cap H^1_0(\mathbb{R}^+)\big)\dotplus\mathrm{span}\{\Phi_\kappa\}\,.
  \end{split}
 \end{equation}
 
To conclude this subsection we prove Theorem \ref{thm:1Fried}.

\begin{proof}[Proof of Theorem \ref{thm:1Fried}]
Since $S-h_0^{(\nu)}$ is bounded, both $h_0^{(\nu)}$ and $S$ have deficiency index one. Parts (i) and (ii) follow at once, respectively from Lemma~\ref{lem:DomOpFridS} and Lemma~\ref{lem:DomSFrid}, since the shift does not modify the domains. Concerning part (iii), it follows from
\[
\Big(h_{0,F}^{-1}+\frac{\nu^2}{4 \kappa^2} \Big)^{-1} \;=\; S_F^{-1} \;=\;R_G
\]
and from the expression \eqref{eq:Green} for the kernel of $R_G$, using the definitions \eqref{eq:FPhi} and 
\eqref{eq:value_of_W}.
\end{proof}

\subsection{Kre{\u\i}n-Vi\v{s}ik-Birman classification of the extensions}~\label{sec:KVB_for_S}


Based on the Kre{\u\i}n-Vi\v{s}ik-Birman extension theory \cite[Theorem 3.4]{GMO-KVB2017}, applied to the present case of deficiency index one, the self-adjoint extensions of $S$ correspond to those restrictions of $S^*$ to subspaces of $\mathcal{D}(S^*)$ that, in terms of formula \eqref{eq:DSstar_firstversion}, are identified by the condition
\begin{equation}\label{eq:c1betac0}
 c_1\;=\;\beta c_0\qquad\textrm{for some }\,\beta\in\mathbb{R}\cup\{\infty\}\,,
\end{equation}
the extension parametrised by $\beta=\infty$ having the domain \eqref{eq:DSF2} and being therefore the Friedrichs extension.

\begin{remark}
 If one replaces the restriction condition \eqref{eq:c1betac0} with the same expression where now $\beta$ is allowed to be a generic complex number, this gives all possible \emph{closed} extensions of $S$ between $\overline{S}$ and $S^*$, as follows by a straightforward application of Grubb's extension theory (see, e.g., \cite[Chapter 13]{Grubb-DistributionsAndOperators-2009}), namely the natural generalisation of the Kre{\u\i}n-Vi\v{s}ik-Birman theory for closed extensions. A recent application of Grubb's theory to operators of point interactions, including $(-\Delta)|_{C^\infty_0(\mathbb{R}^3\setminus\{0\})}$ in $L^2(\mathbb{R}^3)$, from the point of view of Friedrichs systems, is presented in \cite{EM-FriedrichsDelta2017}. 
\end{remark}

Let us denote with $S_\beta$ the extension selected by \eqref{eq:c1betac0} for given $\beta$. Owing to \eqref{eq:DSstar_firstversion} and \eqref{eq:c1betac0}, a generic $g\in\mathcal{D}(S_\beta)$ decomposes as
\begin{equation}\label{eq:SelfAdjointFunction-0}
 g\;=\;f+\beta c_0\Psi_\kappa+c_0\Phi_\kappa
\end{equation}
for unique $f\in H^2_0(\mathbb{R}^+)$ and $c_0\in\mathbb{C}$. The asymptotics \eqref{eq:AsymM0}, \eqref{eq:RGPhi_atzero}, and \eqref{eq:ffprimeDSclosure} imply
\begin{equation}\label{eq:SelfAdjointFunction}
 \begin{split}
  g(r)\;&=\;\frac{c_0}{\Gamma(1-\kappa)}+\frac{c_0\, \nu}{\Gamma(1-\kappa)}\, r \ln r \\
  &\quad + \Big(c_0 \, \nu \, \frac{\,2\psi(1-\kappa)+2 \ln(-\frac{\nu}{\kappa})+(4 \gamma-2) +\kappa^{-1}\,}{2\, \Gamma(1- \kappa)}+c_0\, \beta \,\Gamma(1 - \kappa)  \Vert \Phi_\kappa \Vert^2\Big)\,r \\
  &\quad + o(r^{3/2})\qquad\textrm{as }\,r\downarrow 0\,.
 \end{split}
\end{equation}
The $O(1)$-term and $O(r\ln r)$-term in \eqref{eq:SelfAdjointFunction} come from $\Phi_\kappa$, and so does the first $O(r)$-term; the second $O(r)$-term comes instead from $\Psi_\kappa$; the $o(r^{3/2})$-remainder comes from $f$.

The analogous asymptotics for a generic function $g\in\mathcal{D}(S^*)$ is
\begin{equation}\label{eq:AdjointFunction}
g(r)\;=\; C_0\Big(\frac{1}{\Gamma(1-\kappa)}+\frac{\nu}{\Gamma(1-\kappa)} r \ln r\Big) + C_1 r + o(r^{3/2})\qquad\textrm{as }\,r\downarrow 0
\end{equation}
for some $C_0,C_1\in\mathbb{C}$, as follows again from \eqref{eq:AsymM0}, \eqref{eq:RGPhi_atzero}, and \eqref{eq:ffprimeDSclosure} applied to \eqref{eq:DSstar_firstversion}.
Comparing \eqref{eq:SelfAdjointFunction} with \eqref{eq:AdjointFunction} we conclude the following.

\begin{proposition}[Classification of extensions at $\ell=0$: shift-dependent formulation]\label{prop:classificationThm}~

\noindent The self-adjoint extensions of $S$ form a family $\{S_\beta\,|\,\beta\in\mathbb{R}\cup\{\infty\}\}$. The extension with $\beta=\infty$ is the Friedrichs extension $S_F$. For $\beta\in\mathbb{R}$, the extension $S_\beta$ is the restriction of $S^*$ to the domain $\mathcal{D}(S_\beta)$ that consists of all functions in $\mathcal{D}(S^*)$ for which the coefficient $C_0$ of the leading term $\frac{1}{\Gamma(1-\kappa)}+\frac{\nu}{\Gamma(1-\kappa)} r \ln r$ and the coefficient $C_1$ of the next $O(r)$-subleading term, as $r\downarrow 0$, are constrained by the relation
\begin{equation}\label{eq:bc}
 \frac{C_1}{C_0} \; = \; c_{\nu,\kappa}\, \beta+ d_{\nu,\kappa}\,,
\end{equation}
where
\begin{equation}\label{eq:c-d-coeff}
 \begin{split}
  c_{\nu,\kappa} \;& := \; \Gamma(1 - \kappa)\, \|\Phi_\kappa \|_{L^2}^2\\
  d_{\nu,\kappa} \; &:= \; \nu\,\frac{\, 2 \psi (1- \kappa)+2 \ln(-\frac{\nu}{\kappa})+2(2\gamma - 1)   +\kappa^{-1}\,}{2\, \Gamma(1- \kappa)}\,.
 \end{split}
\end{equation}
Equivalently,
\begin{equation}\label{eq:structural_classification}
 \begin{split}
  S_\beta\;&=\;S^*\upharpoonright\mathcal{D}(S_\beta) \\
  \mathcal{D}(S_\beta)\;&=\;\big\{g=f+\beta c_0\Psi_\kappa+c_0\Phi_\kappa\,\big|\,f\in H^2_0(\mathbb{R}^+),\, c_0\in\mathbb{C} \big\}\,.
 \end{split}
\end{equation}
\end{proposition}

Within the  Kre{\u\i}n-Vi\v{s}ik-Birman extension scheme an equivalent classification in terms of quadratic forms is available. In the present setting, \cite[Theorem 3.6]{GMO-KVB2017} yields at once the following.

\begin{proposition}[Shift-dependent classification at $\ell=0$: form version]\label{prop:form_Sbeta}~

\noindent The self-adjoint extensions of $S$ form a family $\{S_\beta\,|\,\beta\in\mathbb{R}\cup\{\infty\}\}$. The extension with $\beta=\infty$ is the Friedrichs extension $S_F$. For $\beta\in\mathbb{R}$, the extension $S_\beta$ has quadratic form
 \begin{equation}\label{eq:formextensions}
 \begin{split}
  \mathcal{D}[S_\beta]\;&=\;\mathcal{D}[S_F]\dotplus\mathrm{span}\{\Phi_\kappa\} \\
  S_\beta[\phi_\kappa+c_\kappa\Phi_\kappa]\;&=\;S_F[\phi_\kappa]+\beta|c_\kappa|^2\|\Phi_\kappa\|_{L^2}^2
 \end{split}
\end{equation}
for generic $\phi_\kappa\in\mathcal{D}[S_F]$ and $c_\kappa\in\mathbb{C}$. 
\end{proposition}

Thus, the classification provided by Proposition \ref{prop:classificationThm} identifies each extension \emph{directly from the short distance behaviour of the elements of its domain}, and the self-adjointness condition \eqref{eq:bc} is a constrained \emph{boundary condition} as $r\downarrow 0$ (see Remark \ref{rem:previous_literature} below for further comments).
This turns out to be particularly informative for practical purposes, including our next purposes of classification of the discrete spectra of the $S_\beta$'s.

The Friedrichs extension, $\beta=\infty$, is read out from \eqref{eq:bc} as $C_0=0$ and $C_1=c_{\nu,k}$, upon interpreting $C_0\beta=1$. In this case, as expected, \eqref{eq:structural_classification} takes the form of \eqref{eq:DSF2} and \eqref{eq:formextensions} is interpreted as $\mathcal{D}[S_{\beta=\infty}]=\mathcal{D}[S_F]$. Moreover, the following feature of $S_F$ is now obvious from \eqref{eq:structural_classification} and from the short-distance asymptotics of $\Phi_\kappa$ and $\Psi_\kappa$ given by \eqref{eq:AsymM0} and \eqref{eq:RGPhi_atzero} above.

\begin{corollary}\label{cor:Frie}
 The Friedrichs extension $S_F$ is the \emph{only} member of the family $\{S_\beta\,|\,\beta\in\mathbb{R}\cup\{\infty\}\}$ with operator domain contained in $\mathcal{D}[r^{-1}]$, i.e., it is the only self-adjoint extension whose domain's functions have finite expectation of the potential (and hence also of the kinetic) energy. 
\end{corollary}

Another immediate consequence of the extension parametrisation \eqref{eq:structural_classification}, as an application of Kre{\u\i}n's resolvent formula for deficiency index one \cite[Theorem 6.6]{GMO-KVB2017}, is the following.

\begin{corollary}\label{cor:resolventSbeta}
The self-adjoint extension $S_\beta$  is invertible if and only if $\beta\neq 0$, in which case
\begin{equation}
S_\beta^{-1}\; = \; S_F^{-1}+\frac{1}{\beta} \frac{1}{\Vert \Phi_\kappa \Vert^2} |\Phi_\kappa \rangle \langle \Phi_\kappa|\,.
\end{equation}
\end{corollary}

\begin{remark}
 Unlike the Friedrichs extension, the `energy' $S_\beta[g]$ of an element $g\in\mathcal{D}[S_\beta]$ when $\beta\neq\infty$ differs from the formal expression $\|g'\|_{L^2}^2+\nu\|r^{-\frac{1}{2}}g\|\|_{L^2}^2+\eta\|g\|_{L^2}^2$, $\eta=\frac{\nu^2}{4\kappa^2}$. The latter would be instead \emph{infinite} for a generic $g$, and the \emph{finiteness} of $S_\beta[g]$ can be interpreted as the effect of an infinite $\beta$-dependent correction to the above-mentioned formal expression such that the two infinities cancel out. Explicitly, let us write $g=\phi_\kappa+c_\kappa\Phi_\kappa$ as in \eqref{eq:formextensions} and compute
 \[
  \begin{split}
   S_\beta&[g]\;=\;\|\phi_\kappa'\|_{L^2}^2+\nu\|r^{-\frac{1}{2}}\phi_\kappa\|\|_{L^2}^2+\eta\|\phi_\kappa\|_{L^2}^2+\beta|c_\kappa|^2\|\Phi_\kappa\|_{L^2}^2 \\
   &=\;\|g'-c_\kappa\Phi_\kappa'\|_{L^2}^2+\nu\|r^{-\frac{1}{2}}(g-c_\kappa\Phi_\kappa)\|\|_{L^2}^2+\eta\|g-c_\kappa\Phi_\kappa\|_{L^2}^2 +\beta|c_\kappa|^2\|\Phi_\kappa\|_{L^2}^2\,.
  \end{split}
 \]
`Opening the squares' in the above norms clearly yields infinities, so we only proceed formally here, understanding the following expressions as the $\varepsilon\downarrow 0$ limit of integrations that are supported on $(\varepsilon,+\infty)$. One would then have
\[
 \begin{split}
  S_\beta[g]\;&=\;\|g'\|_{L^2}^2+\nu\|r^{-\frac{1}{2}}g\|\|_{L^2}^2+\eta\|g\|_{L^2}^2 \\
  &\qquad +|c_\kappa|^2\Big(-\overline{\Phi_\kappa(0)}\,\Phi_\kappa'(0)+\int_0^{+\infty}\!\!\overline{\Phi_\kappa}\big(-\Phi_\kappa''+{\textstyle\frac{\nu}{r}}\Phi_\kappa+\eta\Phi_\kappa\big)\,\ud r\Big) \\
  &\qquad -2\,\mathfrak{Re}\,c_\kappa\Big(-\overline{g(0)}\,\Phi_\kappa'(0)+\int_0^{+\infty}\!\!\overline{g}\,\big(-\Phi_\kappa''+{\textstyle\frac{\nu}{r}}\Phi_\kappa+\eta\Phi_\kappa\big)\,\ud r\Big) \\
  &\qquad+\beta|c_\kappa|^2\|\Phi_\kappa\|_{L^2}^2\,.
 \end{split}
\]
Using that $-\Phi_\kappa''+{\textstyle\frac{\nu}{r}}\Phi_\kappa+\eta\Phi_\kappa=0$,  $c_\kappa=g(0)/\Phi_\kappa(0)$, and $\Phi_\kappa$ is real-valued, we find
\[
 S_\beta[g]\;=\;\|g'\|_{L^2}^2+\nu\|r^{-\frac{1}{2}}g\|\|_{L^2}^2+\eta\|g\|_{L^2}^2 +|g(0)|^2\Big(\frac{\Phi_\kappa'(0)}{\Phi_\kappa(0)}+\beta\,\frac{\,\|\Phi_\kappa\|_{L^2}^2}{|\Phi_\kappa(0)|^2}\Big).
\]
The $\beta$-dependent correction is now evident from the above expression, that must be interpreted as a \emph{compensation} between the infinite `formal form of $g$' given by the first three summands, and the infinite correction given by the fourth summand -- observe indeed that $\Phi_\kappa'(r)/\Phi_\kappa(r)=(\nu\ln r)(1+o(1))$ as $r\downarrow 0$. Only for the Friedrichs extension this correction is absent and $S_F[g]$ is given by the usual formula.
 \end{remark}

\begin{remark}\label{rem:previous_literature}
 As mentioned in the Introduction, our boundary-condition-driven classification of the self-adjoint realisations of the differential operator $\widetilde{S}$ on the half-line has several precursors in the literature \cite{Rellich-1944,Bulla-Gesztesy-1985}. In fact, the analysis of radial Schr\"{o}dinger operators with Coulomb potentials, and more generally of the so-called `Whittaker operators' $-\frac{\ud^2}{\ud r^2}+(\frac{1}{4}-\mu^2)\frac{1}{\,r^2}-\frac{\kappa}{r}$ on half-line, is also quite active in the present days \cite{Gesztesy-Zinchenko-2006,Bruneau-Derezinski-Georgescu-2011,Derezinski-Richard-2017,DR-2017}.
 The very `spirit' of the structural formula \eqref{eq:structural_classification} is to link, through the extension parameter $\beta$, the `regular' (in this context: rapidly vanishing) behaviour at the origin of the component $f+\beta c_0\Psi_\kappa$ with the `singular' (non-vanishing) behaviour of the component $c_0\Phi_\kappa$ of a generic $g\in\mathcal{D}(S_\beta)$, and the boundary condition of self-adjointness \eqref{eq:bc} is a convenient re-phrasing of that. Lifting the analysis to the three dimensional case makes this terminology more appropriate, as remarked after Theorem \ref{thm:3Dclass}.
\end{remark}

The $\beta$-parametrisation in Propositions \ref{prop:classificationThm} and \ref{prop:form_Sbeta} is shift-dependent and it is convenient now to re-scale $\beta$ so as to re-parametrise the extensions in a shift-independent way. To this aim, for $g\in\mathcal{D}(S^*)$ we set
\begin{equation}\label{eq:g0g1limits}
 \begin{split}
  g_0\;&:=\;\frac{C_0}{\,\Gamma(1-\kappa)}\;=\;\lim_{r\downarrow 0}g(r) \\
  g_1\;&:=\;C_1\;=\;\lim_{r\downarrow 0}r^{-1}\big(g(r)-g_0(1+\nu r\ln r)\big)
 \end{split}
\end{equation}
so that \eqref{eq:AdjointFunction} reads
\begin{equation}
 g\;=\;g_0(1 + \nu\,r\ln r)+g_1 r +o(r^{3/2})\qquad\textrm{as }r\downarrow 0\,,
\end{equation}
and we also define
\begin{equation}\label{eq:alphabeta}
 \alpha\;:=\;\frac{1}{4\pi}\,\Gamma(1-\kappa)\,( c_{\nu,\kappa}\, \beta+ d_{\nu,\kappa} )\,.
\end{equation}
Then, as obvious from \eqref{eq:bc}-\eqref{eq:c-d-coeff},
\begin{equation}\label{eq:DSbeta-alpha}
 \mathcal{D}(S_\beta)\;=\;\{g\in\mathcal{D}(S^*)\,|\,g_1=4\pi\alpha g_0\}\,.
\end{equation}
Moreover, an easy computation applying \eqref{eq:alphabeta} yields
 \begin{equation}\label{eq:beta-alpha-computation}
  \frac{1}{\beta} \frac{1}{\Vert \Phi_\kappa \Vert^2}\;=\;\frac{{\,\Gamma(1-\kappa)^2}}{4\pi}\frac{1}{ \alpha -\frac{\nu}{4\pi}\big(\psi(1-\kappa)+\ln(-\frac{\nu}{\kappa})+(2\gamma-1)+\frac{1}{2 \kappa}\big)}\,.
 \end{equation}

This brings directly to the proof of our main result for the radial problem.

\begin{proof}[Proof of Theorem \ref{thm:1Dclass}] Removing the shift from $S$ to $h_0^{(\nu)}$ does not alter the domain of the corresponding self-adjoint extensions or adjoints, and modifies trivially their action. Thus, part (i) follows from Proposition \ref{prop:classificationThm} and from formulas \eqref{eq:g0g1limits} and   \eqref{eq:DSbeta-alpha} for $\mathcal{D}(S_\beta)$, using the expression \eqref{eq:Sstar_maximal} for $\mathcal{D}(S^*)$, whereas part (ii) follows from Corollary \ref{cor:resolventSbeta} with $S_\beta^{-1}=\big( h_{0,\alpha}^{(\nu)}+\frac{\nu^2}{4 \kappa^2} \big)^{-1}$ and $S_F^{-1}=\big( h_{0,F}^{(\nu)}+\frac{\nu^2}{4 \kappa^2} \big)^{-1}$, together with the identity \eqref{eq:beta-alpha-computation}. So far we have worked with $0<|\kappa|<\frac{1}{2}$: thanks to the uniqueness of the analytic continuation, this determines unambiguously the resolvent at any point in the resolvent set. We can then extend all our previous formulas to the whole regime $(-\infty,0)\cup(0,1)$ for which the expression $\Gamma(1-\kappa)$ still makes sense.
\end{proof}

\subsection{Reconstruction of the 3D hydrogenoid extensions}\label{sec:3Dreconstr}~

Finally, let us re-phrase the previous conclusions in terms of self-adjoint realisations of the hydrogenoid-type operator
\begin{equation}
 \mathring{H}^{(\nu)}\;=\;-\Delta+\frac{\nu}{\,|x|\,}\,,\qquad\mathcal{D}(\mathring{H}^{(\nu)})\;=\;C^\infty_0(\mathbb{R}^3\!\setminus\!\{0\})
\end{equation}
(see \eqref{eq:Hring1} above) on $L^2(\mathbb{R}^3)$.
The self-adjoint extensions of the shifted operator $\mathring{H}^{(\nu)}+\eta{\mathbbm{1}}$, $\eta:=\frac{\nu^2}{4\kappa^2}$, in the sector of angular symmetry $\ell=0$ of $L^2(\mathbb{R}^3)$ are precisely those found in Proposition \ref{prop:classificationThm}.



\begin{proof}[Proof of Theorem \ref{thm:3Dclass}]~

\underline{Part (i)}. Formula \eqref{eq:ang_decomp_Halpha} is obvious from \eqref{eq:ang_decomp}-\eqref{eq:ang_decomp_operator} and from Theorem \ref{thm:1Dclass}.
%
%
 
 \underline{Part (ii)}. Obviously the unique self-adjoint extension of $\mathring{H}^{(\nu)}$, hence necessarily the Friedrichs extension, in the sectors with angular symmetry $\ell\geqslant 1$ is the projection onto such sectors of the operator \eqref{HnuFriedr}, owing to part (i) of this theorem. In the sector $\ell=0$ the operator \eqref{HnuFriedr} acts as $(-\frac{\ud^2}{\ud r^2}+\frac{\nu}{r})\otimes\mathbbm{1}$ and it remains to recognise that its radial domain consists of those $f$'s in $H^2(\mathbb{R}^+)$ that vanish as $f(r)=O(r)$ as $r\downarrow 0$, because this is precisely $\mathcal{D}(S_F)$. This is standard: spherically symmetric elements of $H^2(\mathbb{R}^3)$ are functions $F(|x|)$ for $F\in L^2(\mathbb{R}^+,r^2\,\ud r)$ such that $\Delta_x F\in L^2(\mathbb{R}^3,\ud x)$ and hence $\frac{1}{r^2}\frac{\ud}{\ud r}(r^2\frac{\ud}{\ud r} F)\in L^2(\mathbb{R}^+,r^2\ud r)$; on the other hand $F=\frac{f}{r}$ for $f\in L^2(\mathbb{R}^+,\ud r)$, whence $\frac{1}{r^2}\frac{\ud}{\ud r}(r^2\frac{\ud}{\ud r} F)=\frac{\,f''}{r}$, and the square-integrability of $\Delta_x F$ reads $f''\in L^2(\mathbb{R}^+,\ud r)$; therefore, $f\in H^2(\mathbb{R}^+)$ and $f(r)=rF(r)=O(r)$ as $r\downarrow 0$. Last, the feature mentioned in the statement which identifies uniquely the Friedrichs extension follow from Proposition \ref{prop:classificationThm} and Corollary \ref{cor:Frie}, thanks to the equivalence $\alpha=\infty$ $\Leftrightarrow$ $\beta=\infty$.
 
 \underline{Part (iii)}. Owing to parts (i) and (ii) we only have to establish \eqref{eq:HnuResolvent} over the sector $\ell=0$. In this sector, \emph{radially},
 \[
  (h^{(\nu)}_{0,\alpha}+{\textstyle\frac{\nu^2}{\,4\kappa^2}}\,\mathbbm{1})^{-1}\;=\;(h^{(\nu)}_{0,\infty}+{\textstyle\frac{\nu^2}{\,4\kappa^2}}\,\mathbbm{1})^{-1}+\frac{1}{\beta} \frac{1}{\Vert \Phi_\kappa \Vert^2} |\Phi_\kappa \rangle \langle \Phi_\kappa|
 \]
 owing to Corollary \ref{cor:resolventSbeta}. 
 Formula \eqref{eq:ourgnukappa} reads 
 \[
  \mathfrak{g}_{\nu,k}(x)\;=\;
  \frac{\,\Gamma(1-\kappa)}{4\pi}\,\frac{\Phi_\kappa(|x|)}{|x|}\;=\;\frac{\,\Gamma(1-\kappa)}{\sqrt{4\pi}}\,\frac{\Phi_\kappa(|x|)}{|x|}\otimes Y_0^0
 \]
 and therefore the projection $|\mathfrak{g}_{\nu,k}\rangle\langle\mathfrak{g}_{\nu,k}|$ acting on $L^2(\mathbb{R}^3)$ acts radially in the $\ell=0$ sector as the projection $\frac{\,\Gamma(1-\kappa)^2}{4\pi}|\Phi_\kappa\rangle\langle\Phi_\kappa|$.
 This proves that 
 \[
  \Big(H^{(\nu)}_\alpha+\frac{\nu^2}{\,4\kappa^2}\,\mathbbm{1}\Big)^{\!-1}\;=\;\Big(H^{(\nu)}+\frac{\nu^2}{\,4\kappa^2}\,\mathbbm{1}\Big)^{\!-1} +\frac{1}{\beta} \frac{1}{\Vert \Phi_\kappa \Vert^2}\,\frac{4\pi}{\,\Gamma(1-\kappa)^2}\,|\mathfrak{g}_{\nu,k}\rangle\langle\mathfrak{g}_{\nu,k}|\,.
 \]
 Combining the formula above with \eqref{eq:beta-alpha-computation} finally yields the resolvent formula \eqref{eq:HnuResolvent}.
 
  \underline{Part (iv)}. This is a standard consequence of part (ii) -- see, e.g., the argument in the proof of \cite[Theorems I.1.1.3 and I.2.1.2]{albeverio-solvable}.  
\end{proof}

%

\section{Perturbations of the discrete spectra}

In this Section we prove Theorem \ref{thm:EV_corrections} and we add a few additional observations.

We deliberately choose another path as compared to the standard approach \cite{Zorbas-1980,AGHKS-1983_Coul_plus_delta,Bulla-Gesztesy-1985} that determines the eigenvalues as poles of the resolvent \eqref{eq:HnuResolvent} (see Remark \ref{rem:swavepoles} below), and we exploit instead the radial analysis of extensions that we have developed in Sec.~\ref{sec:Section_of_Classification}. This completes our approach based on the  Kre{\u\i}n-Vi\v{s}ik-Birman extension theory.

\subsection{The $s$-wave eigenvalue problem}\label{sec:swaveEVproblem}~

For fixed $\alpha\in\mathbb{R}$ and $\nu\in\mathbb{R}$ let $\Psi\in\mathcal{D}(H^{(\nu)}_\alpha)$ and $E<0$ satisfy $H^{(\nu)}_\alpha\Psi=E\Psi$ with $\Psi$ belonging to the $L^2$-sector with angular symmetry $\ell=0$. 

In view of \eqref{eq:ang_decomp} we write
\begin{equation}
 \Psi(x)\;=\;\frac{\,g(|x|)\,}{\,\sqrt{4\pi}\,|x|\,}
\end{equation}
for some $g\in\mathcal{D}(S_\beta)\subset L^2(\mathbb{R}^+)$ such that $S_\beta g=(E+\frac{\nu^2}{4\kappa^2})g$, where $\beta$ is given by \eqref{eq:alphabeta} for the chosen $\alpha$ and $\nu$, and a chosen $\kappa\in(0,1)$ (see \eqref{eq:conditions_kappa} above). Thus,
\begin{equation}\label{eq:radEVproblem1}
\Big(\!-\frac{\ud^2}{\ud r^2} +\frac{\nu}{r} \Big) \,g\;=\; E g\,.
\end{equation}

Passing to re-scaled energy $\mathfrak{e}$, radial variable $\rho$, coupling $\vartheta$, and unknown $h$ defined by
\begin{equation}
 \begin{split}
  \mathfrak{e}\;&:=\;\frac{\,4\kappa^2 E}{\nu^2}+1\,,\qquad\qquad\qquad\qquad\!\! \rho\;:=\;-r\,\frac{\,\nu\sqrt{1-\mathfrak{e}}}{\kappa}\;=\;2 r \sqrt{|E|\,}  \\
  \vartheta\;&:=\;\frac{\kappa}{\,\sqrt{1-\mathfrak{e}}\,}\;=\;\frac{-\nu}{\,2\sqrt{|E|}}\,,\qquad u(\rho)\;:=\;g(r)\,,
 \end{split}
\end{equation}
the eigenvalue problem \eqref{eq:radEVproblem1} takes the form
\begin{equation}\label{eq:radEVproblem2}
 \Big(\!-\frac{\ud^2}{\ud \rho^2} -\frac{\vartheta}{\rho}+\frac{1}{4} \Big) \,u\;=\;0\,,
\end{equation}
namely a Whittaker equation of the same type \eqref{eq:hypergeom} above, whose only square-integrable solutions on $\mathbb{R}^+$, analogously to what argued in Section \ref{sec:homo}, are the multiples of Whittaker's function
\begin{equation}
u(\rho)\;=\;\mathscr{W}_{\vartheta,\frac{1}{2}}(\rho) \;=\;e^{-\frac{1}{2}\rho}\, \rho\, U_{1-\vartheta,2}(\rho)\,.
\end{equation}
Therefore, up to multiples, the solution to \eqref{eq:radEVproblem1} is
\begin{equation}
g(r)\;=\; 
\mathscr{W}_{\vartheta,\frac{1}{2}}(2r\sqrt{|E|})\,.
\end{equation}

By means of the expansion \eqref{eq:AsymM0} and of the identity $\nu=-2\sqrt{|E|}\,\vartheta$ one finds
\begin{equation}
 \begin{split}
  g_0\;&=\;\frac{1}{\,\Gamma(1+\frac{\nu}{2\sqrt{|E|}})} \\
	g_1\; &=\; \nu \, \frac{\, \psi\big(1+\frac{\nu}{2\sqrt{|E|}}\big)+ \ln(2 \sqrt{|E|}) +2\gamma - 1 - {\textstyle\frac{\sqrt{|E|}}{\nu}}}{\, \Gamma\big(1+\frac{\nu}{2\sqrt{|E|}}\big)}\,,
 \end{split}
\end{equation}
and such two constants must satisfy the condition $g_1=4\pi\alpha g_0$, as prescribed by Theorem \ref{thm:1Dclass} , because the considered eigenfunction $\Psi$ belongs to $\mathcal{D}(H_\alpha^{(\nu)})$.
We have thus proved that $E$ is an eigenvalue for $H_\alpha^{(\nu)}$ if and only if
\begin{equation}\label{eq:FEalpha}
 \mathfrak{F}_\nu(E)\;=\;\alpha
\end{equation}
with $\mathfrak{F}_\nu$ defined in \eqref{eq:Feigenvalues}.

When $\nu<0$ the function $(-\infty,0)\ni E\mapsto \mathfrak{F}_\nu(E)$ has vertical asymptotes corresponding to non-positive arguments $1+\frac{\nu}{2\sqrt{|E|}}=-q$, $q\in\mathbb{N}_0$, of the digamma function $\psi$, i.e., at the points $E=E_n^{(\nu)}$ defined by
\begin{equation}\label{eq:EnHydr}
 E_n^{(\nu)}\;:=\;-\frac{\nu^2}{4 n^2}\,,\qquad n\,:=\,q+1\,\in\,\mathbb{N}\,.
\end{equation}
The sequence $(E_n)_{n\in\mathbb{N}}$ is increasing and converges to zero. Within each interval $(E_n,E_{n+1})$ the function $E\mapsto \mathfrak{F}_\nu(E)$ is smooth and strictly monotone increasing, and moreover
\[
 \lim_{E\to -\infty}\,\mathfrak{F}_\nu(E)\;=\;-\infty\,.
\]
Thus, for any $\alpha\in\mathbb{R}$ does the equation \eqref{eq:FEalpha} admit countably many negative simple roots, which form the increasing sequence $(E_n^{(\nu,\alpha)})_{n\in\mathbb{N}}$ and accumulate at zero. Therefore, the $s$-wave point spectrum of $H^{(\nu)}_\alpha$ consists precisely of the $E_n^{(\nu,\alpha)}$'s. In the extremal case $\alpha=\infty$ one has $E_n^{(\nu,\alpha=\infty)}=E_n^{(\nu)}$: indeed, the $s$-wave point spectrum of the Friedrichs extension $H^{(\nu)}$ is the ordinary non-relativistic hydrogenoid $s$-wave spectrum, as given by \eqref{eq:EnHydr}.

When $\nu>0$ the function $(-\infty,0)\ni E\mapsto \mathfrak{F}_\nu(E)$ is smooth and strictly monotone increasing, with
\[
 \lim_{E\to-\infty}\,\mathfrak{F}_\nu(E)=-\infty\,,\qquad  \lim_{E\uparrow 0}\mathfrak{F}_\nu(E)\;=\;\frac{\nu}{4\pi}\,(\ln\nu+2\gamma-1)\;=:\;\alpha_\nu\,,
\]
the latter limit following from \eqref{eq:Feigenvalues} owing to the asymptotics \cite[Eq.~(6.3.18)]{Abramowitz-Stegun-1964} that here reads
\[
 \psi\big(1+{\textstyle\frac{\nu}{2\sqrt{|E|}}}\big)\;\stackrel{E\to 0}{=} \;\ln{\textstyle\frac{\nu}{2\sqrt{|E|}}}+O(\sqrt{|E|})\,.
\]
Thus, the equation \eqref{eq:FEalpha} has no negative roots if $\alpha\geqslant\alpha_\nu$ and one negative root if $\alpha<\alpha_\nu$.

This completes the proof of Theorem \ref{thm:EV_corrections}.

\begin{figure}
\includegraphics[width=6cm]{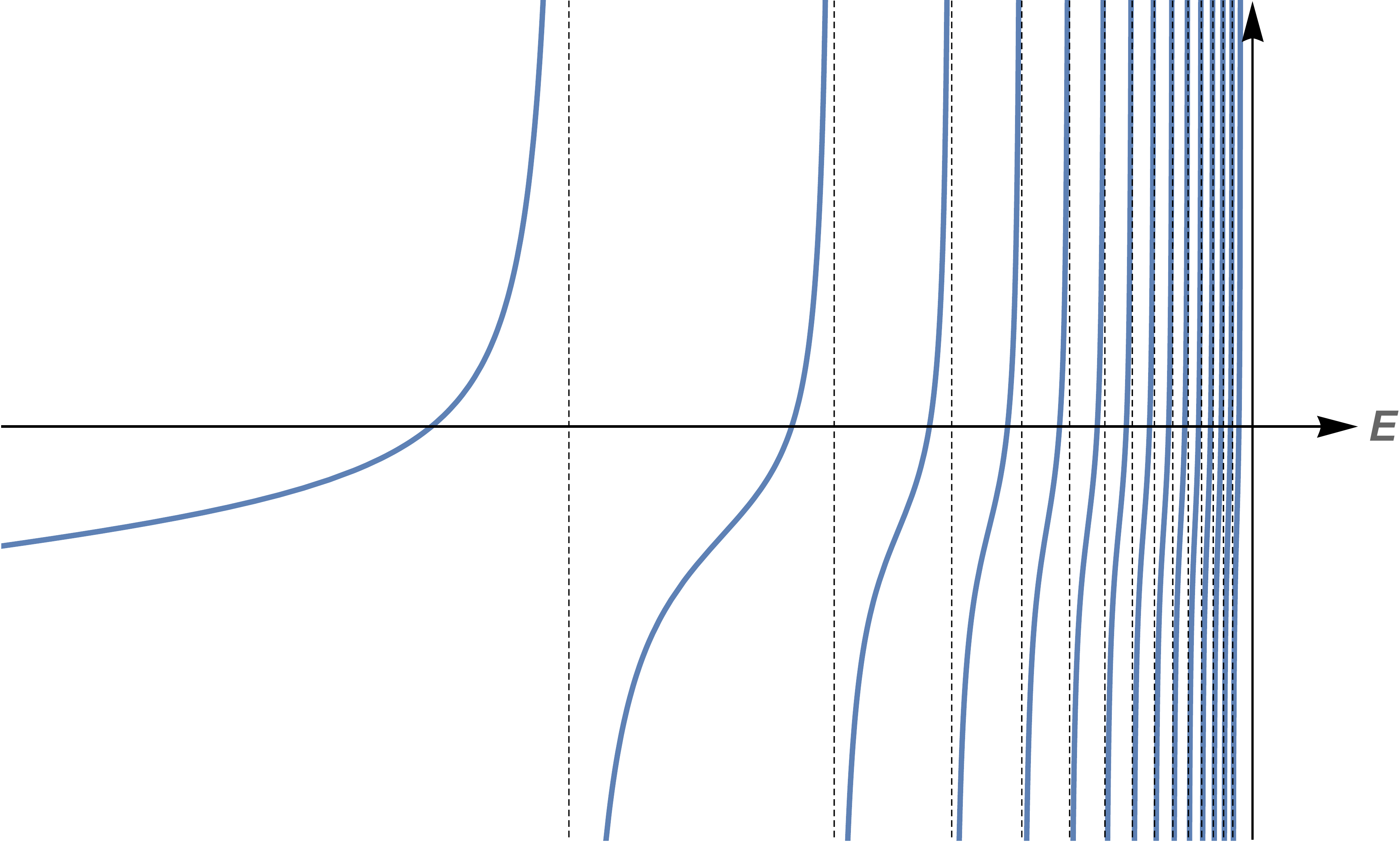}\quad
\includegraphics[width=6cm]{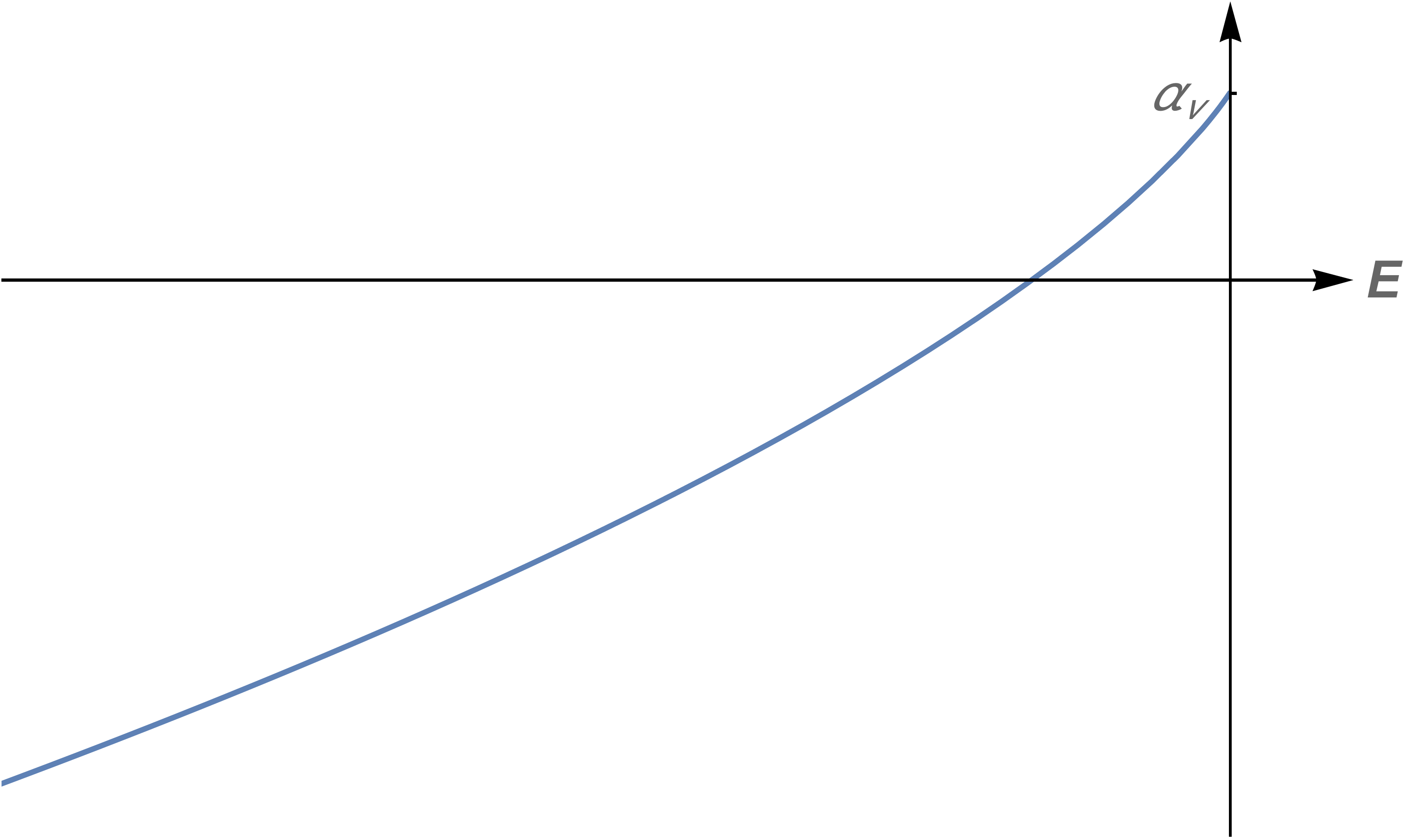}
\caption{Behaviour of the function $(-\infty,0)\ni E\mapsto \mathfrak{F}_\nu(E)$ for $\nu<0$ (left) and $\nu>0$ (right). Dashed lines in the figure represent vertical asymptotes.} \label{fig:Fnu}
\end{figure}

\subsection{Further remarks}\label{sec:EVfurtherremarks}

\begin{remark}\label{rem:EVdecreased}
 The result of Theorem \ref{thm:EV_corrections} when $\nu<0$ confirms that $H^{(\nu)}\geqslant H^{(\nu)}_\alpha$, namely that the Friedrichs extension is larger (in the sense of self-adjoint operator ordering) than any other extension. In particular, $E_{n+1}^{(\nu,\alpha)}\geqslant E_n^{(\nu)}\geqslant E_n^{(\nu,\alpha)}$. 
\end{remark}

\begin{remark}\label{rem:spectra_fibre}
 As is clear from the behaviour of the roots to  $\mathfrak{F}_\nu(E)=\alpha$ (Fig.~\ref{fig:Fnu})
 \begin{equation}
  \bigcup_{\alpha}\sigma_{\mathrm{p}}^{(0)}(H^{(\nu)}_\alpha)\;=\;(-\infty,0)\,.
 \end{equation}
 In this sense the spectra $\sigma_{\mathrm{p}}^{(0)}(H^{(\nu)}_\alpha)$ fibre, as $\alpha$ runs over $\mathbb{R}\cup\{\infty\}$, the whole negative real line.
\end{remark}

\begin{remark}
 When $\nu<0$ one has
 \[
  \lim_{\alpha\to-\infty}E_1^{(\nu,\alpha)}\;=\;-\infty\,,\qquad \lim_{\alpha\to-\infty}E_{n+1}^{(\nu,\alpha)}\;=\;E_{n}^{(\nu)}\,,\quad n=2,3,\dots
 \]
 Both limits are obvious from the behaviour of the function $\mathfrak{F}_\nu(E)$ (Fig.~\ref{fig:Fnu}); the former in particular is a consequence of general facts of the Kre{\u\i}n-Vi\v{s}ik-Birman theory, in the following sense. That there exists only one eigenvalue of $H^{(\nu)}_\alpha$ below the bottom $E_1^{(\nu)}$ of the Friedrichs extension $H^{(\nu)}$ is a consequence of $\mathring{H}^{(\nu)}$ having deficiency index one and of the general result \cite[Corollary 5.10]{GMO-KVB2017}. Moreover, such eigenvalue, which is precisely $E_1^{(\nu,\alpha)}$, must satisfy 
 \[
  E_1^{(\nu,\alpha)}\;\leqslant\;\beta\;=\;\frac{1}{\,c_{\nu,\kappa}}\Big(\frac{4\pi\alpha}{\Gamma(1-\kappa)}-d_{\nu,\kappa}\Big)
 \]
 for any fixed $\kappa\in(0,1)$, as a consequence of \eqref{eq:alphabeta} and of the general property \cite[Theorem 5.9]{GMO-KVB2017}. Thus, the limit $\alpha\to-\infty$ in the above inequality reproduces the limit for $E_1^{(\nu,\alpha)}$.
\end{remark}

 \begin{remark}
  When $\nu>0$ Theorem \ref{thm:EV_corrections} implies that $ H^{(\nu)}_\alpha\geqslant\mathbbm{O}$ if and only if $\alpha\geqslant\alpha_\nu$. This fact too can be understood in terms of a general property of the Kre{\u\i}n-Vi\v{s}ik-Birman theory \cite[Theorem 3.5]{GMO-KVB2017}, which in the present setting reads
  \[
   H^{(\nu)}_\alpha+\frac{\nu^2}{4\kappa^2}\mathbbm{1}\;\geqslant\;\mathbbm{O}\qquad\Leftrightarrow\qquad\beta\,\geqslant\, 0
  \]
  for any $\kappa<0$. The limit $\kappa\to -\infty$ and \eqref{eq:alphabeta} then yields
  \[
   \begin{split}
    H^{(\nu)}_\alpha\;\geqslant\;\mathbbm{O}\qquad\Leftrightarrow\qquad \alpha\,& \geqslant\,\lim_{\kappa\to-\infty}\,\frac{\Gamma(1-\kappa)}{4\pi}\, d_{\nu,\kappa}\\
    &=\,\lim_{\kappa\to-\infty}\,\frac{\nu}{4\pi}\,\big(\psi(1-\kappa)+\ln({\textstyle-\frac{\nu}{\kappa}})+2\gamma-1+{\textstyle\frac{1}{2\kappa}}\big) \\
    &=\: \frac{\nu}{4\pi}\,(\ln\nu+2\gamma-1)\;=\;\alpha_\nu\,,
   \end{split}
  \]
 the limit above following again from the asymptotics \cite[Eq.~(6.3.18)]{Abramowitz-Stegun-1964}.  
 \end{remark}

\begin{remark}\label{rem:swavepoles}
 Having identified the eigenvalues of $H^{(\nu)}_\alpha$ as the roots of $\mathfrak{F}_\nu(E)=\alpha$ is clearly consistent with the fact that such eigenvalues are all the poles  of the resolvent $(H^{(\nu)}_\alpha-z\mathbbm{1})^{-1}$, $z=-\frac{\nu^2}{4\kappa^2}$, determined in \eqref{eq:HnuResolvent}, i.e., the values $E=-\frac{\nu^2}{4\kappa^2}$ with $\kappa$ determined by $\mathfrak{F}_{\nu,\kappa}=\alpha$, which is precisely another way of writing  $\mathfrak{F}_\nu(E)=\alpha$.
 \end{remark}

\bibliographystyle{siam}
\def\cprime{$'$}

\end{document}